%% file: main.tex
\pdfoutput=1

\documentclass[a4paper,UKenglish,cleveref, autoref, thm-restate]{lipics-v2021}

\usepackage[ruled, vlined, linesnumbered]{algorithm2e}
\usepackage{numprint}
\npdecimalsign{.} %
\usepackage{makeidx}
\usepackage[dvipsnames]{xcolor}
\usepackage{color}
\usepackage{hyperref}
\usepackage{xspace}
\usepackage{booktabs}
\usepackage{mathtools}
\usepackage{subcaption}
\usepackage{makecell} 	%
\usepackage{upgreek} %
\usepackage{siunitx}
\usepackage{nicefrac}
\usepackage{multirow}
\usepackage{niceAlgorithm}
\usepackage{cleveref}

\newcommand{\changed}[1]{#1}

\newcommand{\transultra}{Trans-ULTRA\xspace}

\newcommand{\ie}{i.e.,\xspace~}
\newcommand{\eg}{e.g.,\xspace~}

\input{include/notation.tex}
\input{include/colors.tex}
\input{include/figures.tex}

\newcommand{\codestyling}[1]{\texttt{#1}}
\newcommand{\printTime}[3]{\num[minimum-integer-digits = 2]{#1}{:}\num[minimum-integer-digits = 2]{#2}{:}\num[minimum-integer-digits = 2]{#3}}

\title{FLASH-TB: Integrating Arc-Flags and Trip-Based Public Transit Routing}

\titlerunning{FLASH-TB: Integrating Arc-Flags and Trip-Based Public Transit Routing}

\author{Ernestine Großmann}{Heidelberg University, Heidelberg, Germany}{e.grossmann@informatik.uni-heidelberg.de}{https://orcid.org/0000-0002-9678-0253}{}
\author{Jonas Sauer\footnote{Corresponding author}}{University of Bonn, Bonn, Germany}{jsauer1@uni-bonn.de}{https://orcid.org/0000-0002-7196-7468}{}
\author{Christian Schulz}{Heidelberg University, Heidelberg, Germany}{christian.schulz@informatik.uni-heidelberg.de}{https://orcid.org/0000-0002-2823-3506}{}
\author{Patrick Steil}{Karlsruhe Institute of Technology, Karlsruhe, Germany}{patrick@steil.dev}{https://orcid.org/0000-0003-3282-4533}{}
\author{Sascha Witt}{Karlsruhe Institute of Technology, Karlsruhe, Germany}{sascha.witt@kit.edu}{https://orcid.org/0000-0002-7867-3200}{}

\authorrunning{E. Großmann, J. Sauer, C. Schulz, P. Steil and S. Witt}

\Copyright{Ernestine Großmann, Jonas Sauer, Christian Schulz, Patrick Steil, Sascha Witt}



\category{}

\relatedversion{This publication is partially based on an extended abstract that appeared at the~21st International Symposium on Experimental Algorithms~\hbox{\cite{Gro23}}.}

\supplement{Software (Code): \url{https://github.com/TransitRouting/FLASH-TB}}

\acknowledgements{
	We thank Patrick Brosi for providing us with the Germany dataset. 
	We also thank our anonymous reviewers for their extensive feedback, which allowed us to identify and fix critical flaws in the proofs.
	This publication is partially based on an extended abstract that appeared at the~21st International Symposium on Experimental Algorithms~\hbox{\cite{Gro23}}.
}%

\nolinenumbers %

\hideLIPIcs  %

\EventEditors{John Q. Open and Joan R. Access}
\EventNoEds{2}
\EventLongTitle{42nd Conference on Very Important Topics (CVIT 2016)}
\EventShortTitle{CVIT 2016}
\EventAcronym{CVIT}
\EventYear{2016}
\EventDate{December 24--27, 2016}
\EventLocation{Little Whinging, United Kingdom}
\EventLogo{}
\SeriesVolume{42}
\ArticleNo{23}

\begin{document}

\maketitle

\begin{abstract}
	We present FLASH-TB, a journey planning algorithm for public transit networks that combines Trip-Based Public Transit Routing (TB) with the Arc-Flags speedup technique.
	The basic idea is simple:
		The network is partitioned into a configurable number of cells.
		For each cell and each possible transfer between two vehicles, the algorithm precomputes a flag that indicates whether the transfer is required to reach the cell.
		During a query, only flagged transfers are explored.
		Our algorithm improves upon previous attempts to apply Arc-Flags to public transit networks, which saw limited success due to conflicting rules for pruning the search space.
		We show that these rules can be reconciled while still producing correct results.
		Because the number of cells is configurable, FLASH-TB offers a tradeoff between query time and memory consumption.
		It is significantly more space-efficient than existing techniques with a comparable preprocessing time, which store generalized shortest-path trees: to match their query performance, it requires up to two orders of magnitude less memory.
		The fastest configuration of FLASH-TB achieves a speedup of more than two orders of magnitude over TB, offering sub-millisecond query times even on large countrywide networks.
\end{abstract}
\section{Introduction}
	\label{sec:introduction}
	Journey-planning applications that provide real-time routing information have become a part of our daily lives.
	To allow for interactive use, the employed routing algorithms must be very fast.
	Although Dijkstra's algorithm~\cite{Dij59} finds shortest paths in quasi-linear time, it still takes several seconds on continental-sized networks.
	Practical applications therefore rely on \emph{speedup techniques}, which compute auxiliary data in a preprocessing phase to accelerate the query phase.
	For road networks, many speedup techniques have been developed in recent decades~\cite{Bas16b}.
	These achieve sub-millisecond query times with only moderate preprocessing time and space consumption.
	For public transit networks, the situation is less satisfactory.
    Although some techniques offer sub-millisecond query times on networks representing large countries, these require auxiliary data in the tens to hundreds of gigabytes.
	This discrepancy has been explained by the fact that road networks exhibit beneficial structural properties that are less pronounced in public transit networks~\cite{Bas09}.
	Another challenge is that passengers typically consider multiple criteria when choosing between journeys.
	Therefore, most algorithms Pareto-optimize at least two criteria: the arrival time and the number of used trips.
	This increases the solution size and thereby slows down the algorithms even further.
	For these reasons, a speedup technique that achieves very low query times with a moderate amount of precomputed data has remained elusive.
	
	\subsection{Related Work}
	\label{par:stateoftheart}
	For this work, we consider public transit journey planning algorithms that Pareto-optimize arrival time and number of trips.
	For a more general overview, we refer to a survey by~\hbox{\cite{Bas16b}}.
	Classical approaches model the public transit timetable as a graph and apply a multicriteria variant of Dijkstra's algorithm~\cite{Mue07b,Dis08}.
	In the~\emph{time-dependent} graph model~\cite{Bro04,Pyr08}, nodes represent stops in the network and edges are weighted with functions that map departure time to arrival time.
	This yields a compact graph but requires a time-dependent version of Dijkstra's algorithm.
	By contrast, the~\emph{time-expanded} model~\cite{Mue07b,Pyr08} introduces nodes for each event in the timetable (\eg a vehicle arriving at or departing from a stop).
	Edges connect consecutive events of the same trip and events between which a transfer is possible.
    This results in a larger graph, but Dijkstra's algorithm can be applied directly.

    Graph-based models can be combined with speedup techniques for Dijkstra's algorithm.
	A notable example is Arc-Flags~\cite{Hil09,Lau09,Moe06}, which partitions the graph into cells and computes a boolean \emph{flag} for each combination of edge and cell.
    The flag indicates whether the edge is required to reach the cell.
	Arc-Flags has been applied to both time-dependent~\hbox{\cite{Ber09}} and time-expanded~\hbox{\cite{Del09c}} graphs, although only the arrival time was optimized in the latter case.
    However, this has only produced speedups of~3--4, compared to over~5\,000~\hbox{\cite{Bas16b}} for road networks.
	
	More recent algorithms do not model the timetable as a graph.
	Instead, they use tailor-made data structures that allow for cache-efficient queries.
	Notable examples are RAPTOR~\cite{Del15b} and Trip-Based Routing~(TB)~\hbox{\cite{Wit15}}.
	The latter employs a lightweight preprocessing phase that computes relevant transfers between pairs of trips.
	This yields query times below~100\,ms, even on large networks, significantly improving upon graph-based techniques.

	Algorithms that offer query times in the sub-millisecond range do so by precomputing auxiliary data whose size is quadratic in the size of the network.
	Public Transit Labeling~(PTL)~\hbox{\cite{Del15}} adapts the ideas of Hub Labeling~\hbox{\cite{Coh03}} to time-expanded graphs.
	On metropolitan networks, PTL achieves query times of a few microseconds but requires tens of gigabytes of space.
    Additional space in the hundreds of gigabytes would be required to support~\emph{journey unpacking}, \ie retrieving full descriptions of the optimal journeys.
	Transfer Patterns~(TP)~\hbox{\cite{Bas10}} is a technique that, in its most basic form, answers all possible queries during the preprocessing phase.
    The journey descriptions are then condensed into \emph{transfer patterns}, \ie sequences of stops at which transfers between vehicles occur.
	This yields a search graph for each possible source stop, which is explored during the query phase.
	On the network of Germany, TP answers queries in less than~1\,ms but requires hundreds of hours of preprocessing time and over 100\,GB of space~\hbox{\cite{Bas16}}.

    To make the preprocessing phase more practical, \cite{Bas10} split the transfer patterns into global components between a small set of important \emph{hub stations} and local components that lead to and from the hub stations.
	This significantly reduces the memory consumption, but the preprocessing time is improved only when applying a heuristic that causes a small percentage of queries to be answered incorrectly.
	A similar approach that answers all queries correctly is Scalable Transfer Patterns~\hbox{\cite{Bas16}}, which employs graph clustering to split the patterns into local and global components.
	However, when these are recombined during a query, the search space increases drastically.
	Thus, the resulting query times are only barely competitive with TB.
	Trip-Based Routing Using Condensed Search Trees~(TB-CST)~\cite{Wit16} re-engineers the ideas of TP with a faster, TB-based preprocessing algorithm.
	To save space, shared suffixes are extracted from the search graphs and re-attached by the query algorithm.
	TB-CST offers faster query times than TP on country-scale networks, although it is less successful on metropolitan areas.
	
	\subsection{Contribution}
	We present FLASH-TB (TB with FLAgged SHortcuts), a speedup technique for TB that exceeds the query performance of TP and TB-CST with a much smaller memory footprint.
    Like those techniques, FLASH-TB answers all possible queries in advance, but it condenses the journey descriptions in a different manner.
    Instead of computing transfer patterns, FLASH-TB stores and annotates the individual transfers.
	The simplest variant stores a flag for every combination of transfer and target stop $\targetStop$ that indicates whether the transfer is needed to reach $\targetStop$.
	Because transfers connect events rather than stops, this approach retains fine-grained information that is lost in the transfer patterns.
    The memory consumption can be reduced by leveraging Arc-Flags: the stops are partitioned into $k$ cells and a combined flag is stored for each pair of transfer and target cell.
    Thus, FLASH-TB is conceptually simpler than previous space-saving techniques for TP.
    Furthermore, because $k$ is configurable, it offers a tradeoff between query performance and space consumption.
	
	Our approach is similar to applying Arc-Flags to the time-expanded graph, but there are crucial differences that allow it to overcome the issues observed by~\cite{Del09c}.
    In particular, we observe that the low speedups that were previously reported are caused by a conflict between different pruning rules in the preprocessing and query phases.
	We resolve this conflict with a redesigned preprocessing step and show that this vastly reduces the search space.
	We perform an experimental evaluation on four real-world datasets representing metropolitan and country-scale public transit networks.
	With less than~1\,GB of precomputed data, FLASH-TB achieves query times of~500\,$\mu s$ on the network of Germany and well below~100\,$\mu s$ on smaller country-scale networks.
	This corresponds to a speedup of two orders of magnitude over TB.
	The metropolitan network of Paris requires~8\,GB to achieve a similar speedup.
	When allowed to use as much space as TB-CST, FLASH-TB is faster by a factor of two to nine.
	
	\subsection{Outline}
	This paper is structured as follows:
	Section~\ref{sec:preliminaries} introduces the necessary notation.
	Section~\ref{sec:tb} outlines the algorithms upon which our work is based, including TB, TB-CST, and Arc-Flags.
	In Section~\ref{sec:transferset}, we discuss challenges that arise when using TB in the preprocessing step, which is done by both TB-CST and FLASH-TB, and how they can be overcome while keeping the search space small.
	Based on this, we present FLASH-TB in Section~\ref{sec:flashtb} and prove its correctness.
	We discuss similarities and differences to existing speedup techniques in Section~\ref{sec:comparison}.
	Section~\ref{sec:experiments} evaluates the performance of FLASH-TB on real-world networks and compares it with state-of-the-art algorithms.
	Finally, we summarize our findings and discuss potential areas for future research in Section~\ref{sec:conclusion}.
	
	\section{Preliminaries}
	\label{sec:preliminaries}
	Our definitions and notation are partially based on~\cite{Wit15} and~\cite{Bau23}.
	
	\subsection{Public Transit Network}
	A~\emph{public transit network} is a 5-tuple $(\stops,\stopEvents,\trips,\lines,\footpaths)$ consisting of a set of stops $\stops$, a set of stop events $\stopEvents$, a set of trips $\trips$, a set of lines $\lines$, and a set of footpaths $\footpaths\subseteq\stops\times\stops$.
	A \textit{stop} $\aStop\in\stops$ is a location visited by a vehicle at which passengers can embark or disembark.
	A \textit{stop event} $\stopEvent\in\stopEvents$ represents a visit of a vehicle at the stop $\aStop(\stopEvent) \in \stops$ with the arrival time $\arrivalTime(\stopEvent)$ and the departure time $\departureTime(\stopEvent) \geq \arrivalTime(\stopEvent)$.
	A \textit{trip} is a sequence $\aTrip = \left< \stopEvent_0, \stopEvent_1, \dots\right>$ of stop events performed by the same vehicle.
	We denote the $i$-th stop event of $\aTrip$ as $\aTrip[i]$ and the number of stop events in $\aTrip$ as $\absoluteVal{\aTrip}$.
	The \textit{stop sequence} of $\aTrip$ is given by $\left< \aStop(\stopEvent_0), \aStop(\stopEvent_1), \dots\right>$.
	We say that a trip $\tripA$ \textit{overtakes} another trip $\tripB$ with the same stop sequence if $\arrivalTime(\tripA[0]) \geq \arrivalTime(\tripB[0])$ and if there is an index $i$ such that $\arrivalTime(\tripA[i]) \leq \arrivalTime(\tripB[i])$ or $\departureTime(\tripA[i]) \leq \departureTime(\tripB[i])$.
	A line $\aLine\in\lines$ is a set of trips with the same stop sequence such that no trip overtakes another.
	Every trip $\aTrip$ is included in exactly one line, which is denoted by $\aLine(\aTrip)$.
	We define a total order on trips $\tripA, \tripB$ of the same line $\aLine(\tripA)=\aLine(\tripB)$:
	\begin{align*}
		\tripA \prec \tripB &\iff \arrivalTime(\tripA[0]) < \arrivalTime(\tripB[0]),\\
		\tripA \preceq \tripB &\iff \tripA = \tripB \lor \tripA \prec \tripB.
	\end{align*}
	If $\tripA \prec \tripB$, it follows from the non-overtaking property that $\arrivalTime(\tripA[i]) < \arrivalTime(\tripB[i])$ and $\departureTime(\tripA[i]) < \departureTime(\tripB[i])$ for $0 \leq i < \absoluteVal{\aTrip}$.
	Furthermore, if two trips with the same stop sequence have the same arrival or departure time at any point along the sequence, our definition of overtaking forces them to be in different lines.
	For a trip $\aTrip$, its predecessor according to $\prec$ is denoted by $\pred{\aTrip}$.
	If $\aTrip$ is the first trip of its line, we write $\pred{\aTrip}=\bot$.
	A partial order of stop events $\tripA[i],\tripB[j]$ is given by
	\begin{align*}
		\tripA[i] \preceq \tripB[j] &\iff \aLine(\tripA)=\aLine(\tripB) \land \tripA \preceq \tripB \land i \leq j,\\
		\tripA[i] \prec \tripB[j] &\iff \tripA[i] \preceq \tripB[j] \land (\tripA \prec \tripB \lor i < j).
	\end{align*}
	
	A~\emph{footpath} $(\stopA,\stopB)\in\footpaths$ allows passengers to transfer between stops $\stopA$ and $\stopB$, which requires the time $\transfertime{\stopA}{\stopB}$.
	If no footpath between $\stopA$ and $\stopB$ exists, we define $\transfertime{\stopA}{\stopB} = \infty$.
	For each stop $\aStop\in\stops$, we require that a footpath $(\aStop,\aStop)\in\footpaths$ exists with $\transfertime{\aStop}{\aStop} = 0$; we call this an \emph{empty} footpath.
	We require that the set of footpaths is transitively closed and fulfills the triangle inequality, \ie if there are stops $\stopA,\stopB,\stopC\in\stops$ with $(\stopA,\stopB)\in\footpaths$ and $(\stopB,\stopC)\in\footpaths$, then there must be a footpath $(\stopA,\stopC)\in\footpaths$ with $\transfertime{\stopA}{\stopC}\leq\transfertime{\stopA}{\stopB}+\transfertime{\stopB}{\stopC}$.
	
	\subsection{Journeys}
	A~\textit{trip segment} $\tripSegment{\aTrip}{i}{j}$ is the subsequence of the trip $\aTrip$ between the two stop events $\aTrip[i]$ and $\aTrip[j]$.
	A \textit{transfer} is a tuple $\transfer{\tripA[i]}{\tripB[j]}\in\stopEvents\times\stopEvents$ with $\arrivalTime(\tripA[i])+\transfertime{\aStop(\tripA[i])}{\aStop(\tripB[j])}\leq\departureTime(\tripB[j])$.
	It represents a passenger changing from trip $\tripA$ to $\tripB$ at the respective stop events.
	For a source stop $\sourceStop$ and target stop $\targetStop$, a $\sourceStop$-$\targetStop$-\textit{journey} $\aJourney = \left< \footpath_0, \tripSegment{\aTrip_1}{i_1}{j_1}, \dots, \tripSegment{\aTrip_k}{i_k}{j_k}, \footpath_{k+1} \right>$ is a sequence of trip segments such that a transfer connects every pair of consecutive trip segments.
	In addition, the journey contains the \textit{initial footpath} $\footpath_0=(\sourceStop, \aStop(\aTrip_1[i_1]))\in\footpaths$, which connects $\sourceStop$ to the first trip, and the \textit{final footpath} $\footpath_{k+1}=(\aStop(\aTrip_k[j_k]), \targetStop)\in\footpaths$ connecting the last trip to $\targetStop$.
	For the sake of brevity, we omit empty footpaths from the notation.
    We denote the source stop of $\aJourney$ by $\sourceStop(\aJourney)$ and the target stop by $\targetStop(\aJourney)$.
    The number of trips used by $\aJourney$ is denoted as $\absoluteVal{\aJourney}=k$.
	The arrival time of $\aJourney$ at $\targetStop$ is given by $\arrivalTime(\aJourney)=\arrivalTime(\aTrip_k[j_k])+\transfertime{\aStop(\aTrip_k[j_k])}{\targetStop}$.
	The \textit{departure time} of $\aJourney$ is the latest possible departure time at $\sourceStop$ such that $\aTrip_1$ can still be entered, \ie $\departureTime(\aJourney)=\departureTime(\aTrip_1[i_1])-\transfertime{\sourceStop}{\aStop(\aTrip_1[i_1])}$.	
	
	The~\emph{stop sequence} $\stops(\aJourney)$ of $\aJourney$ is obtained by removing every stop from the sequence $\widetilde{\stops}(\aJourney)=\left<\sourceStop,\aStop(\aTrip_1[i_1]),\aStop(\aTrip_1[j_1]),\dots,\aStop(\aTrip_k[i_k]),\aStop(\aTrip_k[j_k]),\targetStop\right>$ that is identical to its predecessor.
    \changed{Note that if a stop $\aStop$ appears twice in $\stops(\aJourney)$, then there is a corresponding journey $\aJourney'$ that waits at $\aStop$ instead of taking the loop.
    For the criteria considered in this work (arrival time and number of trips), $\aJourney'$ is always preferable to $\aJourney$.
    Hence, we assume throughout this work that each stop in $\stops(\aJourney)$ appears only once.}
	For two stops $\stopA,\stopB\in\stops(\aJourney)$, the \textit{subjourney} $\subjourney{\stopA}{\stopB}$ is the $\stopA$-$\stopB$-journey whose trip segments form a subsequence of $\aJourney$.
	A \emph{partial journey} $\aJourney_p = \left< \footpath_0, \tripSegment{\aTrip_1}{i_1}{j_1}, \dots, \tripSegment{\aTrip_k}{i_k}{j_k}, \aTrip_{k+1}[i_{k+1}] \right>$ differs from a journey in that the final footpath is replaced by a transfer to the \emph{target event} $\targetEvent(\aJourney_p)=\aTrip_{k+1}[i_{k+1}]$.
	The trip $\aTrip_{k+1}$, which is entered but not exited, is still counted as a used trip, \ie $\absoluteVal{\aJourney_p}=k+1$.
	
	We define three types of \emph{prefixes} of a journey $\aJourney = \left< \footpath_0, \tripSegment{\aTrip_1}{i_1}{j_1}, \dots, \tripSegment{\aTrip_k}{i_k}{j_k}, \footpath_{k+1} \right>$:
	\begin{itemize}
		\item For $0 < \ell \leq k$, the $\sourceStop$-$\aStop(\aTrip_\ell[j_\ell])$-journey $\left< \footpath_0, \tripSegment{\aTrip_1}{i_1}{j_1}, \dots, \tripSegment{\aTrip_\ell}{i_\ell}{j_\ell}\right>$ is an \emph{exit prefix}.
		\item \changed{For $1 \leq \ell < k$, the $\sourceStop$-$\aStop(\aTrip_{\ell+1}[i_{\ell+1}])$-journey $\left< \footpath_0, \tripSegment{\aTrip_1}{i_1}{j_1}, \dots, \tripSegment{\aTrip_\ell}{i_\ell}{j_\ell},\footpath_{\ell+1}\right>$ with final footpath $\footpath_{\ell+1}=(\aStop(\aTrip_\ell[j_\ell]),\aStop(\aTrip_{\ell+1}[i_{\ell+1}]))\in\footpaths$ is a \emph{footpath prefix}. For $\ell=0$, the $\sourceStop$-$\aStop(\aTrip_1[i_1])$-journey $\langle\footpath_0\rangle$ is also a footpath prefix.}
		\item For $0 \leq \ell < k$, the partial journey $\left< \footpath_0, \tripSegment{\aTrip_1}{i_1}{j_1}, \dots, \tripSegment{\aTrip_\ell}{i_\ell}{j_\ell}, \aTrip_{\ell+1}[i_{\ell+1}]\right>$ is called a \emph{partial prefix}.
	\end{itemize}
	Exit and footpath prefixes are called \emph{proper} prefixes to distinguish them from the partial prefix.
	
	\subsection{Journey Planning Problem}
	A \emph{query} $\query=(\sourceStop,\targetStop,\departureTime)$ consists of source and target stops $\sourceStop,\targetStop\in\stops$ and a departure time $\departureTime$.
	The set $\journeys(\query)$ of \emph{feasible journeys} for $\query$ consists of all $\sourceStop$-$\targetStop$-journeys $\aJourney$ with $\departureTime(\aJourney)\geq\departureTime$.
	A \emph{partial query} $\query_p=(\sourceStop,\targetEvent,\departureTime)$ has a target event $\targetEvent\in\stopEvents$ instead of a target stop.
    The set $\journeys_p(\query_p)$ of \emph{feasible partial journeys} consists of all partial $\sourceStop$-$\targetEvent'$-journeys $\aJourney_p$ with $\departureTime(\aJourney_p)\geq\departureTime$ and $\targetEvent'\preceq\targetEvent$.
	
	Let $\journeys$ denote a set of feasible (partial) journeys.
	A \emph{cost function} $\cost \colon \journeys \to \costspace$ maps (partial) journeys to a \emph{cost space} $\costspace$.
	A \emph{(2D) cost system} $\costsystem = ((\cost_1,\preceq_1), (\cost_2,\preceq_2))$ is formed by cost functions $\cost_1\colon \journeys \to \costspace_1$ and $\cost_2\colon \journeys \to \costspace_2$ with partial orders $\preceq_1,\preceq_2$ on $\costspace_1$ and $\costspace_2$, respectively.
	The~\emph{cost vector} of a (partial) journey $\aJourney$ is the tuple $\cost(\aJourney)=(\cost_1(\aJourney), \cost_2(\aJourney))$ in the 2D cost space $\costspace_1 \times \costspace_2$.
	We call $\aJourney$ a \emph{representative} of $\cost(\aJourney)$.
	Two (partial) journeys $\aJourney,\aJourney'$ are \emph{equivalent} if $\cost(\aJourney)=\cost(\aJourney')$.
	The \emph{dominance relation} of $\costsystem$ is a partial order on $\costspace_1 \times \costspace_2$ that is defined as follows for cost vectors $\cost=(\cost_1,\cost_2)$ and $\cost'=(\cost_1',\cost_2')$:
	\begin{align*}
		\cost \preceq \cost' &\iff \cost_1 \preceq_1 \cost_1' \land \cost_2 \preceq_2 \cost_2',\\
		\cost \prec \cost' &\iff \cost \preceq \cost' \land \cost \neq \cost'.
	\end{align*}
	We say that $\cost$ \emph{weakly dominates} $\cost'$ if $\cost \preceq \cost'$ and $\cost$ \emph{strictly dominates} $\cost'$ if $\cost \prec \cost'$.
	The dominance relation carries over to (partial) journeys: we say $\aJourney$ weakly dominates $\aJourney'$ ($\aJourney \preceq \aJourney'$) if $\cost(\aJourney) \preceq \cost(\aJourney')$ and $\aJourney$ strictly dominates $\aJourney'$ ($\aJourney \prec \aJourney'$) if $\cost(\aJourney) \prec \cost(\aJourney')$.
	A (partial) journey $\aJourney\in\journeys$ and its cost vector $\cost(\aJourney)$ are called \emph{Pareto-optimal} if there is no (partial) journey $\aJourney'\in\journeys$ with $\cost(\aJourney') \prec \cost(\aJourney)$.
	The \changed{\emph{Pareto front}} is the set of all cost vectors $\cost(\aJourney)$ for $\aJourney\in\journeys$ such that $\cost(\aJourney)$ is Pareto-optimal.
	A \emph{representative set} is a subset of $\journeys$ that consists of exactly one representative for every Pareto-optimal cost vector.
	This set is not necessarily unique because there may be multiple representatives for the same cost vector.
    Note that journeys with a non-empty final footpath may dominate journeys with an empty final footpath; this requires that the footpath set is transitively closed.
	
	Unless otherwise noted, we use the cost system $\costsystem=((\arrivalTime,\leq),(\absoluteVal{\cdot},\leq))$ to evaluate journeys and $\costsystem_p=((\targetEvent,\preceq),(\absoluteVal{\cdot},\leq))$ to evaluate partial journeys.
	A journey is Pareto-optimal for a query $\query=(\sourceStop,\targetStop,\departureTime)$ if it is Pareto-optimal among the feasible journeys $\journeys(\query)$.
	Likewise, a partial journey is Pareto-optimal for a partial query $\query_p=(\sourceStop,\targetEvent,\departureTime)$ if it is Pareto-optimal among the feasible partial journeys $\journeys_p(\query_p)$.
	We call $\aJourney$ \emph{exit-optimal} for $\query$ if every exit prefix $\aJourney'$ of $\aJourney$ is Pareto-optimal for the query $(\sourceStop,\targetStop(\aJourney'),\departureTime)$.
	We call $\aJourney$ \emph{prefix-optimal} if every proper prefix $\aJourney'$ is Pareto-optimal for $(\sourceStop,\targetStop(\aJourney'),\departureTime)$ and every partial prefix $\aJourney_p$ for $(\sourceStop,\targetEvent(\aJourney_p),\departureTime)$.
	
	We consider several variants of the \emph{journey planning problem} in a public transit network $(\stops,\stopEvents,\trips,\lines,\footpaths)$.
	Given a query $\query=(\sourceStop,\targetStop,\departureTime)$, the (one-to-one) \textit{fixed departure time problem} asks for \changed{the Pareto front} (according to $\costsystem$) and a representative set.
	For the \textit{profile problem}, we are given an interval $\left[\atime_1, \atime_2\right],\,\atime_1 \leq \atime_2$ of possible departure times in addition to $\sourceStop$ and $\targetStop$.
	Here, the objective is to find the union of the \changed{Pareto fronts} for each query $(\sourceStop,\targetStop,\atime)$ with $\atime \in \left[ \atime_1, \atime_2\right]$, as well as a representative for each Pareto-optimal cost vector.
	In the~\emph{full-range profile problem}, the departure time interval spans the entire service duration of the network.
	In an~\emph{one-to-all query}, no target stop is given.
	Instead, the objective is to solve the fixed departure time or profile problem for every possible target stop $\targetStop\in\stops$.
	Finally, an~\emph{all-to-all query} asks for a solution to the journey planning problem for each possible pair of source and target stops.
	
	\subsection{Graph}
	A directed, weighted graph $\graph = (\vertices, \edges, \edgeWeight)$ consists of a set of nodes $\vertices$, a set of edges $\edges \subseteq \vertices \times \vertices$, and an edge weight function $\edgeWeight: \edges \to \mathbb{R}$.
	A \textit{path} $\aPath = \left< \aVertex_1, \aVertex_2, \dots, \aVertex_k\right>$ is a sequence of nodes between $\aVertex_1$ and $\aVertex_k$ such that an edge connects each pair of consecutive nodes.
	The weight of a path is the sum of the weights of all edges in the path.
	A path $\aPath = \left< \sourceVertex, \dots, \targetVertex \right>$ between a source node $\sourceVertex$ and a target node $\targetVertex$ is called a \textit{shortest path} if there is no path between $\sourceVertex$ and $\targetVertex$ with a smaller weight.
	
	Given a value $\numCells\in\mathbb{N}$ and a graph $\graph = \left(\vertices, \edges, \edgeWeight\right)$, a (\emph{$\numCells$-way}) \emph{partition} of $\graph$ is a function $\partition : \vertices \to \left\{1, \dots, \numCells\right\}$ that partitions the node set $\vertices$ into $\numCells$ \textit{cells}.
	The set of nodes in cell $i$ is denoted as $\vertices_i \coloneqq \partition^{-1}(i)$.
	An edge $(\vertexA,\vertexB) \in \edges$ is called a~\emph{cut edge} if $\vertexA$ and $\vertexB$ belong to different cells.
	A node is called a~\emph{boundary node} if it is incident to a cut edge.
	The partition is called~\emph{balanced} for an imbalance parameter $\varepsilon > 0$ if the size of each cell $\vertices_i$ is bounded by $\absoluteVal{\vertices_i} \leq \left( 1 + \varepsilon \right) \left\lceil \absoluteVal{\vertices}/\numCells\right\rceil$.
	The \emph{graph partitioning problem} asks for a balanced partition that minimizes the weighted sum of all~cut~edges.
	
	\section{Basic Algorithms}
	\label{sec:tb}
	In this section, we discuss the algorithms upon which our work is based.
	We outline the basic TB algorithm in Section~\ref{sec:tb:basic}, extensions for profile and one-to-all search in Section~\ref{sec:tb:extensions}, the TB-CST speedup technique in Section~\ref{sec:tb:cst}, and Arc-Flags in Section~\ref{sec:arcflags}.
	
	\subsection{Trip-Based Routing}
	\label{sec:tb:basic}
	
	The TB~\hbox{\cite{Wit15}} algorithm is split into two phases: a precomputation step generates a set $\transfers\subseteq\stopEvents\times\stopEvents$ of transfers between stop events, which is used in the query phase.
	
	\subparagraph*{Transfer Precomputation.}
	The preprocessing step begins by generating possible transfers between pairs of stop events.
	Then, pruning rules are applied to discard some transfers that are not required to answer queries.
	These pruning rules are not exhaustive, so $\transfers$ may still contain transfers that do not occur in any Pareto-optimal journey.
	The generation step creates all transfers $\aTransfer = \transfer{\tripA[i]}{\tripB[j]}$ with $i > 0$, $j < \absoluteVal{\tripB} - 1$ and $\tripA[i]\not\preceq\tripB[j]$ such that $\tripB$ is the earliest trip of its line that can be entered at $\aStop(\tripB[j])$ when arriving at $\arrivalTime(\tripA[i])+\transfertime{\aStop(\tripA[i])}{\aStop(\tripB[j])}$.
	The \textit{U-turn} reduction rule removes a transfer $\aTransfer = \transfer{\tripA[i]}{\tripB[j]}$ if $\aStop(\tripA[i-1]) = \aStop(\tripB[j+1])$ and $\arrivalTime(\tripA[i-1]) \leq \arrivalTime(\tripB[j+1])$.
	A second reduction rule, which we call the \textit{latest-exit} rule, removes transfers that can be replaced with another transfer from the same trip.
	For every journey of the form $\aJourney = \left<\tripSegment{\tripA}{k}{i}, \tripSegment{\tripB}{j}{\ell} \right>$ that uses the transfer $\aTransfer$, the algorithm checks whether there is a journey of the form $\aJourney'=\left<\tripSegment{\tripA}{k}{i'},\tripSegment{\tripC}{j'}{\ell'}, \footpath \right>$ with $i < i'$, $\footpath=(\aStop(\tripC[\ell']), \aStop(\tripB[\ell]))$ and $\aJourney'\preceq\aJourney$.
	Note that $\aJourney'$ is required to start with the same trip as $\aJourney$ and must exit it at a later stop event.
	This means that many possible journeys that might dominate $\aJourney$ are not considered.
	
	\subparagraph*{Query Algorithm.}
	\begin{figure}
		\begin{minipage}{\textwidth}
			\begin{algorithm}[H]
				\caption{TB trip scanning operation (for target stop $\targetStop$).}
				\label{alg:tbts}
				\myproc{\Scan{$n, \minTime$}\label{alg:tb:ts:begin}}{
					$\aQueue_{n+1}\leftarrow{}\emptyset$\;
					\ForEach{$\tripSegment{\aTrip}{j}{k}\in\aQueue_n$\label{alg:tb:ts:queue:firstloop}}{
						\For{$i$ from $j$ to $k$\label{alg:tb:ts:loop1begin}}{
							\lIfComment{target pruning}{$\arrivalTime(\aTrip[i])\geq\minTime$\label{alg:tb:ts:final}}{\Break}\label{alg:tb:targetpruning}
							\IfComment{create cost vector}{$\arrivalTime(\aTrip[i])+\transfertime{\aStop(\aTrip[i])}{\targetStop}<\minTime$\label{alg:tb:improv}}{
								$\minTime\leftarrow\arrivalTime(\aTrip[i])+\transfertime{\aStop(\aTrip[i])}{\targetStop}$\;
								$\paretoset\leftarrow{}\paretoset\cup\left\{\left(\minTime,n\right)\right\}$, removing dominated entries\;\label{alg:tb:domination}
							}
						}
					}
					\ForEach{$\tripSegment{\aTrip}{j}{k}\in\aQueue_n$\label{alg:tb:ts:queue:secondloop}}{
						\For{$i$ from $j$ to $k$\label{alg:tb:ts:loop2begin}}{
							\lIfComment{target pruning}{$\arrivalTime\left(\aTrip[i]\right)\geq\minTime$}{\Break}\label{alg:tb:targetpruning2}
							\ForEachComment{scan transfers}{$\transfer{\aTrip[i]}{\aTrip'[i']}\in\transfers$}{
								$\Enqueue\left(\aTrip', i'+1, n+1\right)$\label{alg:tb:ts:enqueue}
							}
						}
					}
				}
			\end{algorithm}
		\end{minipage}
	\end{figure}
	
	The TB query algorithm resembles a breadth-first search on a graph with the trips as nodes and the precomputed transfers as edges.
	The algorithm tracks which parts of the network have already been explored by maintaining a~\emph{reached index} $\reachedIndex\left(\aTrip\right)$ for each trip $\aTrip$.
	This is the index of the first reached stop event of $\aTrip$, or $\absoluteVal{\aTrip}$ if none have been reached yet.
	Additionally, the algorithm maintains \changed{the Pareto front $\paretoset$} among the cost vectors of all journeys to $\targetStop$ found so far.
	The search operates in~\emph{rounds}, where round $n$ finds \changed{the journeys in the representative set with $n$ trips}.
	Each round $n$ maintains a first-in-first-out (FIFO) queue $\aQueue_n$ of newly reached trip segments, which are scanned during the round.

    Before the first round, the queue $\aQueue_1$ is filled by examining the lines that visit the source stop $\sourceStop$ or a stop that is reachable from $\sourceStop$ with an initial footpath.
	For each line visiting a reachable stop $\aStop$ with index $i$, the algorithm identifies the earliest trip $\aTrip$ of the line that can be entered at $\aStop$ when arriving there at $\departureTime+\transfertime{\sourceStop}{\aStop}$.
	The $\Enqueue$ operation is then called to test whether a new trip segment starting at $\aTrip[i+1]$ needs to be enqueued.
    Algorithm~\ref{alg:tbenq} depicts the $\Enqueue$ operation of a stop event $\aTrip[j]$.
	If $j<\reachedIndex(\aTrip)$, then the trip segment $\tripSegment{\aTrip}{j}{\reachedIndex(\aTrip)-1}$ is added to the queue for the next round.
	Additionally, for every succeeding trip $\aTrip'\succeq\aTrip$ of the same line, the reached index $\reachedIndex\left(\aTrip'\right)$ is set to $\min\left(\reachedIndex(\aTrip'), j\right)$.
	This ensures that the search only enters the earliest reachable trip of each line, a principle we call \emph{line pruning}.

    Pseudocode for the trip scanning operation is given in Algorithm~\ref{alg:tbts}.
	It omits steps that are needed to retrieve the representative set of journeys; these are discussed below.
	A trip segment $\tripSegment{\aTrip}{j}{k}$ is scanned by iterating over the stop events $\aTrip[i]$ with $j\leq{}i\leq{}k$.
	For each stop event $\aTrip[i]$, the loop in lines~\ref{alg:tb:ts:loop1begin}--\ref{alg:tb:domination} creates a new journey $\aJourney$ by exiting the trip and (if necessary) taking a final footpath to $\targetStop$.
	To discard dominated journeys, the algorithm maintains the earliest arrival time $\minTime$ at the target stop found so far.
	If $\aJourney$ arrives at or after $\minTime$, then it is weakly dominated by another journey, so it is discarded.
	Otherwise, the cost vector $\cost(\aJourney)$ is added to the \changed{Pareto front} $\paretoset$, dominated cost vectors are removed, and $\minTime$ is set to $\arrivalTime(\aJourney)$.
	If the arrival time of $\aTrip[i]$ is not earlier than $\minTime$, then all following stop events can be skipped because they cannot be extended to journeys that improve $\minTime$ (see line~\ref{alg:tb:targetpruning}).
	This principle is known as \textit{target pruning}.
	The loop in lines~\ref{alg:tb:ts:loop2begin}--\ref{alg:tb:ts:enqueue} scans the outgoing transfers of $\aTrip[i]$ to find newly reachable trip segments.
	Again, target pruning is applied in line~\ref{alg:tb:targetpruning2}.
	A transfer $\transfer{\aTrip[i]}{\aTrip'[i']}\in\transfers$ is scanned by calling the $\Enqueue$ operation for $\aTrip'[i'+1]$.

	\begin{figure}
		\begin{minipage}{\textwidth}
			\begin{algorithm}[H]
				\caption{TB enqueuing operation.}
				\label{alg:tbenq}
				\myproc{\Enqueue{$\aTrip, j, n$}}{
					\lIfComment{trip segment already reached}{$\reachedIndex\left(\aTrip\right)\leq{}j$}{\Return}
					$\aQueue_n\leftarrow\aQueue_n\cup\left\{\tripSegment{\aTrip}{j}{\reachedIndex(\aTrip)-1}\right\}$\;
					\ForEachComment{update reached index}{$\aTrip'\succeq\aTrip$}{
						\lIf{$\reachedIndex\left(\aTrip'\right)\leq{}j$}{\Break}
						$\reachedIndex\left(\aTrip'\right)\leftarrow{}j$\;
					}
				}
			\end{algorithm}
		\end{minipage}
	\end{figure}
		
    Note that Algorithm~\ref{alg:tbts} loops over the trip segments in $\aQueue_n$ twice: once to create cost vectors and update $\minTime$, and then again to scan transfers.
	  Using two loops instead of one improves the cache efficiency~\cite{Wit15}.
    It also makes the target pruning check in line~\ref{alg:tb:targetpruning2} more effective because the value of $\minTime$ is final for this round after this first loop.
		
	\subparagraph*{Journey Unpacking.}
	As presented thus far, the query algorithm only computes the \changed{Pareto front $\paretoset$}.
	To retrieve the representative set of journeys, the trip segments in the FIFO queues and the cost vectors additionally store~\emph{parent pointers}.
	For each trip segment $\tripSegment{\aTrip_n}{j_n}{k_n}$ inside a queue $\aQueue_n$, the algorithm stores a pointer to the preceding trip segment $\tripSegment{\aTrip_{n-1}}{j_{n-1}}{k_{n-1}}$, which is located in the previous queue $\aQueue_{n-1}$.
	Additionally, it stores the index $i_{n-1}$ with $j_{n-1} \leq i_{n-1} \leq k_{n-1}$ at which $\aTrip_{n-1}$ is exited when the transfer $\transfer{\aTrip_{n-1}[i_{n-1}]}{\aTrip_n[j_n-1]}$ is used.
	These two values are set during the $\Enqueue$ operation for $\aTrip_n[j_n]$.
	Similarly, each cost vector $(\atime,n)\in\paretoset$ is associated with a pointer to the last trip segment of the corresponding journey, which is located in $\aQueue_n$, and the index at which the trip segment is exited.
	Note that this scheme requires that the FIFO queue $\aQueue_n$ for each round $n$ is kept in memory after the round is done, because it will be accessed again during journey unpacking.
	See~\cite{Wit15} for details on how to implement the queues efficiently.
	
	\subsection{TB Extensions}
	\label{sec:tb:extensions}
	\subparagraph*{Profile Search.}
	Profile-TB extends TB to solve the profile problem.
	It collects all possible departure times at $\sourceStop$ within the departure time interval $[\atime_1,\atime_2]$ and processes them in descending order.
	A departure time $\departureTime$ is possible if there is a stop event $\aTrip[i]$ such that $\departureTime+\transfertime{\sourceStop}{\aStop(\aTrip[i])}=\departureTime(\aStop[i])$.
	For each departure time, a run of the TB query algorithm is performed.
	All data structures, including reached indices, are not reset between runs.
	This exploits the fact that journeys found in a run with departure time $\departureTime$ are still feasible for all departure times before $\departureTime$, so they can be used to prune suboptimal results in subsequent runs.
	This principle is called \emph{self-pruning}.
	To obtain correct results, the definition of reached indices must be modified slightly.
	For each trip $\aTrip$ and each number of trips $n$, the algorithm now maintains a reached index $\reachedIndex_n(\aTrip)$, which is the index of the first stop event in $\aTrip$ that was reached with $n$ or fewer trips.
	Whenever $\reachedIndex_n(\aTrip)$ is updated to $\min(\reachedIndex_n(\aTrip),k)$ for some value $k$, the same is done for the reached indices $\reachedIndex_m(\aTrip)$ with $m \geq n$.
	
	\subparagraph*{One-to-All Search.}
	\label{sec:tb:extensions:one-to-all}
	A one-to-all extension of Profile-TB was briefly outlined by~\cite{Wit16}.
	The same extension can also be applied to the original TB algorithm for fixed departure time queries.
	We give a more detailed description.
	The earliest arrival time $\minTime$ and the \changed{Pareto front} $\paretoset$ are replaced with an earliest arrival time $\arrivalTime(\aStop,n)$ for each stop $\aStop$ and number of trips $n$.
	\changed{This is the minimum arrival time among feasible journeys for the query~$(\sourceStop, \aStop, \departureTime)$ with at most~$n$ trips found so far, or $\infty$ if none have been found (yet)}.
	The \changed{representative} set of journeys with target stop $\aStop$ is \changed{maintained} implicitly: the cost vector $(\arrivalTime(\aStop,n),n)$ is Pareto-optimal if $\arrivalTime(\aStop,n) < \arrivalTime(\aStop,n-1)$, where we define $\arrivalTime(\aStop,-1)=\infty$ for the sake of simplicity.
	To unpack the journeys, the algorithm stores a trip segment pointer and exit index along with each earliest arrival time $\arrivalTime(\aStop,n)$, in the same manner as the one-to-one algorithm.
	As with the reached indices, the profile query does not reset the earliest arrival times between runs.
	
	The trip scanning operation for a round $n$ (cf.\ Algorithm~\ref{alg:tbts}) is modified as follows:
	When a stop event $\aTrip[k]$ is scanned during the first loop \changed{(lines~\ref{alg:tb:ts:queue:firstloop}--\ref{alg:tb:domination})}, the algorithm iterates over all stops $\aStop$~with $\transfertime{\aStop(\aTrip[k])}{\aStop}<\infty$ and computes the arrival time $\overline{\arrivalTime}=\arrivalTime(\aTrip[k])+\transfertime{\aStop(\aTrip[k])}{\aStop}$ at $\aStop$.
	If $\overline{\arrivalTime}<\arrivalTime(\aStop,n)$, then $\arrivalTime(\aStop,m)$ is set to $\min(\overline{\arrivalTime},\arrivalTime(\aStop,m))$ for all $m\geq{}n$.
	The target pruning rule, which is no longer applicable, is replaced with a \textit{local pruning} rule:
	When a stop event $\aTrip[k]$ is scanned during the second loop \changed{(lines~\ref{alg:tb:ts:queue:secondloop}--\ref{alg:tb:ts:enqueue})}, the algorithm checks whether $\arrivalTime(\aTrip[k]) > \arrivalTime(\aStop(\aTrip[k]),n)$.
	If so, the outgoing transfers of $\aTrip[k]$ are not scanned.
	Similar rules are also employed in other algorithms, such as RAPTOR~\cite{Del15b}.
	
	\subsection{Condensed Search Trees}
	\label{sec:tb:cst}
	Trip-Based Routing Using Condensed Search Trees~(TB-CST)~\cite{Wit16} uses TB to precompute search graphs in a similar manner to TP.
	The preprocessing phase solves the all-to-all full-range profile problem by running one-to-all Profile-TB from every stop.
	Consider the Profile-TB search for a source stop $\sourceStop$.
	After each TB run, all newly found representative journeys are unpacked.
	This yields a breadth-first search tree with $\sourceStop$ as the root, trip segments as inner nodes, the reached target stops as leaves, and footpaths and transfers as edges.
	The search trees of all runs are combined into the \emph{prefix tree} of $\sourceStop$.
	Here, each trip segment $\tripSegment{\aTrip}{i}{j}$ is replaced with a tuple $(\aLine(\aTrip),i)$ consisting of the corresponding line $\aLine(\aTrip)$ and the stop index $i$ where the line is entered.
	
	To answer a query $(\sourceStop,\targetStop,\departureTime)$, TB-CST builds a~\emph{query graph} from the prefix tree of $\sourceStop$ by extracting all paths that lead to a leaf representing $\targetStop$.
	Then a variant of Dijkstra's algorithm is run on the query graph.
	Because the prefix tree only provides information about lines, the used trips are reconstructed during the query.
	When exploring an edge from $\sourceStop$ to a tuple $(\aLine,i)$, the algorithm computes the earliest trip of $\aLine$ that can be entered at the $i$-th stop.
	When exploring an edge between tuples $(\aLine,i)$ and $(\aLine',j)$, the used trip $\aTrip$ of $\aLine$ is already known, so the algorithm scans the outgoing transfers of $\aTrip[i]$ in $\transfers$ to find the earliest reachable trip $\aTrip'$ of $\aLine'$.\looseness=-1
	
	The space required to store all prefix trees can be reduced by extracting \emph{postfix trees}.
	Consider the prefix tree for a source stop $\sourceStop$.
	For each path from the root to a leaf representing a target stop $\targetStop$, a~\emph{cut node} is chosen.
	The subpath from the cut node to the leaf is then removed from the prefix tree of $\sourceStop$ and added to the postfix tree of $\targetStop$.
	Because many of these extracted subpaths occur in multiple prefix trees, moving them into a shared postfix tree considerably reduces memory consumption.
	To construct the query graph for a source stop $\sourceStop$ and target stop $\targetStop$, the prefix tree $\sourceStop$ and the postfix tree of $\targetStop$ are spliced back together at the cut nodes.
	
	\subsection{Arc-Flags}
	\label{sec:arcflags}
	Arc-Flags~\cite{Hil09,Lau09,Moe06} is a speedup technique for Dijkstra's algorithm in road networks.
	Given a weighted graph $\graph=\left(\vertices, \edges, \edgeWeight\right)$, the preprocessing phase of Arc-Flags performs two steps:
	First, it computes a partition $\partition: \vertices \to \{1,\dots,\numCells\}$ of the node set into $\numCells$ cells, where $\numCells$ is a freely chosen parameter.
	Then, a~\emph{flags function} $\aFlag : \edges\times\left\{1, \dots, \numCells\right\} \to \left\{0, 1\right\}$ is computed.
	For an edge $\anEdge$ and a cell $i$, the value $\aFlag(\anEdge, i)$ is called a \emph{flag}, and we say that $\anEdge$ is~\emph{flagged} for cell $i$ if $\aFlag(\anEdge, i)=1$.
	The flags function must have the following property: for each pair of source node $\sourceVertex$ and target node $\targetVertex$, there is at least one shortest path $\aPath$ from $\sourceVertex$ to $\targetVertex$ such that every edge $\anEdge$ along $\aPath$ is flagged for the cell $\partition(\targetVertex)$ containing the target node.
	With this precomputed information, a shortest path query between $\sourceVertex$ and $\targetVertex$ can be answered by running Dijkstra's algorithm but only scanning edges that are flagged for $\partition(\targetVertex)$.
	The parameter $\numCells$ imposes a tradeoff between query speed and memory consumption.
	The space required to store the flags is in $\Theta\left(\numCells\absoluteVal{\edges}\right)$, which is manageable for $\numCells \ll \absoluteVal{\vertices}$.
	On the other hand, the search space of the query decreases for larger values of $\numCells$ because fewer edges will be flagged if the target cell is smaller.
	
	Flags can be computed naively by computing a shortest-path tree from every node.
	This can be improved by exploiting the observation that every shortest path that leads into a cell $i$ must pass through a boundary node of $i$.
	Thus, it is sufficient to compute backward shortest-path trees from all boundary nodes.
	For more details, we refer to~\cite{Hil09}.
	
	\section{Precomputing Auxiliary Data with TB}
	\label{sec:transferset}
	TB-CST precomputes its auxiliary data by solving the all-to-all full-range profile problem with Profile-TB.
	Our own approach, which we present in Section~\ref{sec:flashtb}, uses the same precomputation scheme but extracts a different set of data from the computed journeys.
	In this section, we discuss challenges that arise when using Profile-TB for this computation and how they can be overcome.
	
	\subsection{Issues}
	\label{sec:transferset:issues}
    If one-to-all TB with local pruning (see Section~\ref{sec:tb:extensions:one-to-all}) is combined with the original TB transfer precomputation, it may answer some queries incorrectly.
	This is because local pruning is incompatible with the latest-exit reduction rule employed by the preprocessing.
	Local pruning ensures that the journeys found by one-to-all TB are exit-optimal, \ie a trip is only exited if the resulting arrival time at the stop is optimal.
	By contrast, the latest-exit rule prefers journeys that exit a trip at the latest possible stop.
	As shown in Figure~\ref{fig:latestexit}, these two choices can conflict.
	In this example, exiting the trip $\tripA$ at $\tripA[1]$ and walking to $\aStop(\tripA[2])$ is faster than exiting at $\tripA[2]$.
	Therefore, local pruning prevents the transfer $\transfer{\tripA[2]}{\tripB[0]}$ from being scanned.
	However, the latest-exit rule discards the transfer $\transfer{\tripA[1]}{\tripB[0]}$ in favor of $\transfer{\tripA[2]}{\tripB[0]}$.
	As a result, one-to-all TB will not find any journeys from $\aStop(\tripA[0])$ to $\aStop(\tripB[1])$.
	A similar issue can also occur in circle lines that visit the same stop twice.
	Due to local pruning, transfers will only be scanned for the first visit.
	However, the latest-exit rule will only keep the outgoing transfers for the second visit.
	Note that this issue does not occur in one-to-one variants of TB, which do not use local pruning.
	The TB-CST preprocessing as originally implemented by~\hbox{\cite{Wit15}} also avoids this issue because it does not apply local pruning.
	
	\begin{figure}[tp!]
		\caption{
			An example network in which local pruning conflicts with the latest-exit rule.
			Gray boxes represent stops.
			Nodes within the boxes represent stop events and are labeled with their indices along the respective trip.
			Colored edges represent trips, which are labeled with the departure and arrival times of the respective stop events.
			Gray edges represent footpaths (between stops) and transfers (between stop events).
			In this example, it is faster to exit ${\tripA}$ at ${\tripA[1]}$ and take the footpath to ${\aStop(\tripA[2])}$ than to stay seated until ${\tripA[2]}$.
			Therefore, local pruning prevents the one-to-all TB query from scanning the transfer ${\transfer{\tripA[2]}{\tripB[0]}}$.
			However, the latest-exit rule discards the transfer ${\transfer{\tripA[1]}{\tripB[0]}}$.
		}
		\label{fig:latestexit}
		\centering
		\input{fig/LatestExit}
	\end{figure}
	
	Another issue that does occur in TB-CST but previously went unnoticed is caused by a conflict between the line pruning rule employed in the query phase and the self-pruning rule used by Profile-TB.
	As discussed in Section~\ref{sec:tb:cst}, the query algorithm only explores the earliest reachable trip of each line; later trips of the same line are not explored.
	This creates a problem in the example shown in Figure~\ref{fig:conflictline}.
	The one-to-all Profile-TB precomputation from $\sourceStop$ will add an edge from $(\aLine(\tripA),1)$ to $(\aLine(\tripC),0)$ to the prefix tree of $\sourceStop$, but due to self-pruning, it will not add an edge from $(\aLine(\tripA),1)$ to $(\aLine(\tripD),0)$.
	During a query from $\sourceStop$ for the departure time $0$, the algorithm explores the edge from $(\aLine(\tripA),1)$ to $(\aLine(\tripC),0)$.
	However, it only searches for outgoing transfers from the earliest reachable trip of the line, which is $\tripB$.
	No corresponding transfer exists and the transfer from the later trip $\tripA$ is ignored.
	Hence, the algorithm fails to find any journey to $\targetStop$.
	A simple workaround is to disable line pruning in the query algorithm: in addition to scanning the outgoing transfers of the earliest reachable trip of a line, those of all later trips are scanned as well.
	In the example from Figure~\ref{fig:conflictline}, this ensures that the transfer $\transfer{\tripA[1]}{\tripC[0]}$ is scanned and $\tripC$ is reached.
	
	\begin{figure}[tp!]
		\caption{
			An example network showing the conflict between line pruning and self-pruning.
			The trips ${\tripA}$ and ${\tripB}$ belong to the same line.
			For a query from ${\sourceStop}$ to ${\targetStop}$ with departure time ${0}$, a query algorithm with line pruning finds the journey ${\left<\tripSegment{\tripB}{0}{1}, \tripSegment{\tripD}{0}{1}\right>}$.
			The later trip ${\tripA}$ is not entered.
			A Profile-TB query, however, only finds the journey ${\left<\tripSegment{\tripA}{0}{1}, \tripSegment{\tripC}{0}{1}\right>}$ due to self-pruning.
		}
		\label{fig:conflictline}
		\centering
		\input{fig/departureTimeBuffering}
	\end{figure}
	
	\subsection{Property-Preserving Transfer Sets}
	Both issues are caused by the fact that the different phases of the algorithm make inconsistent choices between representatives of the same cost vector.
	Although this can be avoided by disabling certain pruning rules, doing so increases the query search space.
	Instead, we modify the transfer generation and profile search phases to make them compatible with each other and the query phase.

	To formalize the notion of compatibility, we study whether a transfer set preserves certain journey properties.
	For a transfer set $\transfers\subseteq\stopEvents\times\stopEvents$ and a query $\query$, let $\journeys(\query,\transfers)\subseteq\journeys(\query)$ denote the set of feasible journeys whose transfers are all included in $\transfers$.
	Let $x$ denote a property of journeys, such as exit optimality.
	The transfer set $\transfers$ is called \emph{$x$-preserving} for $\query$ if there exists a representative set $\representativeSet$ for $\query$ such that every journey $\aJourney\in\representativeSet$ has the property $x$ and is contained in $\journeys(\query,\transfers)$.
	If $\transfers$ is $x$-preserving for all possible queries, we simply call it $x$-preserving overall.
	
	We show that one-to-all TB with local pruning is correct when using an exit-optimality-preserving~(EOP) transfer set.
	Note that while the set of all possible transfers is EOP, it is extremely large and will lead to slow queries.
	In Section~\ref{sec:transferset:ultra}, we present an algorithm that computes a sufficiently small EOP transfer set.
	First, we show that for each Pareto-optimal cost vector, there is an exit-optimal representative whose stop events are all scanned by the search.
	Note that this is not necessarily the journey returned by the algorithm.
	
	\begin{lemma}
		\label{th:one-to-all:scan}
		Let $\transfers$ be a transfer set and $\aJourney=\left< \footpath_0, \tripSegment{\aTrip_1}{i_1}{j_1}, \dots, \tripSegment{\aTrip_k}{i_k}{j_k}, \footpath_{k+1} \right> \in \journeys(\query,\transfers)$ an exit-optimal journey for a query $\query=(\sourceStop,\targetStop,\departureTime)$.
		Consider a one-to-all TB search for $\query$ using $\transfers$ as the transfer set.
		For $1 \leq n \leq k$, round $n$ scans the stop event $\aTrip_n[j_n]$.
	\end{lemma}
	\begin{proof}
		Because the exit prefix $\left< \footpath_0, \tripSegment{\aTrip_1}{i_1}{j_1}, \dots, \tripSegment{\aTrip_n}{i_n}{j_n} \right>$ is Pareto-optimal, $\arrivalTime(\aStop(\aTrip_n[j_n]),n-1)>\arrivalTime(\aStop(\aTrip_n[j_n]),n)=\arrivalTime(\aTrip_n[j_n])$ must hold after round $n$.
		For a trip $\aTrip$, let $\reachedIndex_n(\aTrip)$ denote the reached index of $\aTrip$ at the start of round $n$.
		The trip scanning operation ensures that $\arrivalTime(\aStop(\aTrip[i]),n) \leq \arrivalTime(\aTrip[i])$ holds after round $n$ for every $\reachedIndex_n(\aTrip) \leq i < \absoluteVal{\aTrip}$.
		Together, this implies $\reachedIndex_{n-1}(\aTrip_n)>j_n$.
		Additionally, because $\arrivalTime(\aTrip'[j_n])<\arrivalTime(\aTrip_n[j_n])$ holds for every preceding trip $\aTrip'\prec\aTrip_n$, it follows that $\reachedIndex_n(\aTrip')>j_n$.
		Hence, the only way that $\reachedIndex_n(\aTrip_n)\leq j_n$ can hold is if round $n-1$ adds a trip segment containing $\aTrip_n[j_n]$ to $\aQueue_n$ and round $n$ scans it.
		We show by induction over $n$ that round $n-1$ calls the \Enqueue operation (see Algorithm~\ref{alg:tbenq}) for $\aTrip_n[i_n+1]$.
		This ensures that $\reachedIndex_n(\aTrip_n) \leq i_n+1 \leq j_n$ holds afterwards, from which it follows that round $n$ scans $\aTrip_n[j_n]$.
		\begin{description}
			\item[\textbf{Base case ${\left(n=1\right)}$:}]
			Because the exit prefix $\left< \footpath_0, \tripSegment{\aTrip_1}{i_1}{j_1} \right>$ is optimal, $\aTrip_1$ is the earliest trip of its line that can be entered at $\aStop(\aTrip_1[i_1])$.
			Hence, round $0$ calls the \Enqueue operation for $\aTrip_1[i_1+1]$.
			\item[\textbf{Inductive step ${\left(n \rightarrow n+1\right)}$:}] Because the exit prefix $\left< \footpath_0, \tripSegment{\aTrip_1}{i_1}{j_1}, \dots, \tripSegment{\aTrip_n}{i_n}{j_n} \right>$ is optimal, it follows that $\arrivalTime(\aStop(\aTrip_n[j_n]),n)\geq\arrivalTime(\aTrip_n[j_n])$.
			Hence, the trip scanning operation of round $n$ does not apply local pruning at $\aTrip_n[j_n]$.
			When the outgoing transfer $\transfer{\aTrip_n[j_n]}{\aTrip_{n+1}[i_{n+1}]}\in\transfers$ is scanned, the \Enqueue operation is called for $\aTrip_{n+1}[i_{n+1}+1]$.
		\end{description}
	\end{proof}
	
	From this, it follows that one-to-all TB computes \changed{the Pareto front}.
	
	\begin{theorem}
		\label{th:one-to-all:correct}
		Let $\transfers$ be an EOP transfer set for a query $\query=(\sourceStop,\targetStop,\departureTime)$.
		A one-to-all TB search for $\query$ using $\transfers$ as the transfer set computes \changed{the Pareto front}.
	\end{theorem}
	\begin{proof}
		Because $\transfers$ is EOP for $\query$, there is a representative set $\representativeSet\subseteq\journeys(\query,\transfers)$ of exit-optimal journeys.
		Consider a journey $\aJourney\in\representativeSet$ with $\absoluteVal{\aJourney}=k$.
		For $1 \leq n \leq k$, let $\tripSegment{\aTrip_n}{i_n}{j_n}$ denote the $n$-th trip segment of $\aJourney$.
		By Lemma~\ref{th:one-to-all:scan}, round $n$ of the one-to-all TB search scans the stop event $\aTrip_n[j_n]$.
		When the outgoing footpaths of $\aTrip_k[j_k]$ are scanned, it ensures that $\arrivalTime(\targetStop,k)=\arrivalTime(\aJourney)$.
	\end{proof}
	
	\begin{corollary}
		\label{th:one-to-one}
		Let $\transfers$ be an EOP transfer set for a query $\query=(\sourceStop,\targetStop,\departureTime)$.
		A one-to-one TB search for $\query$ using $\transfers$ as the transfer set computes \changed{the Pareto front}.
	\end{corollary}
	\begin{proof}
		Analogous to Theorem~\ref{th:one-to-all:correct}.
		Note that the inductive step in Lemma~\ref{th:one-to-all:scan} still holds for one-to-one TB because it does not apply local pruning.
	\end{proof}
	
	Theorem~\ref{th:one-to-all:correct} carries over in a straightforward manner to one-to-all Profile-TB.
	However, it only makes a claim about the \changed{Pareto front}, not the computed representative set of journeys.
	Even with an EOP transfer set, the journeys computed by one-to-all Profile-TB are not necessarily exit-optimal.
	This is due to the self-pruning rule, as outlined in Section~\ref{sec:transferset:issues}.

	\subsection{Transitive ULTRA}
	\label{sec:transferset:ultra}
	An alternative transfer generation algorithm for TB was previously presented in the context of multimodal journey planning.
	Here, the transitively closed set of footpaths is replaced with an unrestricted~\emph{transfer graph} that may represent one or several road-based transportation modes.
	ULTRA~(UnLimited TRAnsfers)~\cite{Bau23} is a technique that enables TB to handle multimodal networks by condensing the transfer graph into a set of edges between pairs of stop events, \ie a transfer set.
	The method for generating this transfer set is substantially different from the TB transfer precomputation and does not employ the latest-exit rule.
    We propose \transultra, a variant of ULTRA for transitively closed footpaths.
	We show that the computed transfer set preserves a journey property called canonicity, which is stronger than exit optimality.
	Because \transultra is identical to ULTRA expect for minor differences in the exploration of footpaths, we only give details that are necessary to prove the canonicity-preserving property.
	For a full overview of ULTRA, including algorithmic details and the original proofs, we refer to~\hbox{\cite{Bau23}}.

	\subparagraph*{Canonical Journeys.}
	We call a journey with exactly two trips and empty initial and final footpaths a \emph{candidate}.
	For example, the candidates contained in the journey $\aJourney = \left< \footpath_0, \tripSegment{\aTrip_1}{i_1}{j_1}, \dots, \tripSegment{\aTrip_k}{i_k}{j_k}, \footpath_{k+1} \right>$ are $\left< \tripSegment{\aTrip_\ell}{i_\ell}{j_\ell}, \tripSegment{\aTrip_{\ell+1}}{i_{\ell+1}}{j_{\ell+1}}\right>$ for $1 \leq \ell < k$.
	ULTRA exploits the observation that a set of transfers that is sufficient for answering all queries can be found by only inspecting candidates.
	Formally, the set of transfers that occur in at least one Pareto-optimal candidate \changed{for at least one possible query} is Pareto-optimality-preserving.
	However, this set is unnecessarily large: it covers all Pareto-optimal journeys, but one representative for every Pareto-optimal cost vector is sufficient.
	
	To choose the representative, ULTRA defines a unique~\emph{tiebreaking sequence} for each journey.
	For \transultra, we modify this sequence slightly to account for the fact that footpaths consist of single edges.
	Let $\lineIndex:\lines\to\mathbb{N}$ and $\stopIndex:\stops\to\mathbb{N}$ be total orderings on the sets of lines and stops.
	Let $\eventIndex\colon\stopEvents\to\mathbb{N}$ be a total ordering of the stop events such that $\eventIndex(\aTrip[i]) < \eventIndex(\aTrip'[j])$ if and only if~\changed{$\langle \lineIndex(\aLine(\aTrip)), \arrivalTime(\aTrip[0]), i \rangle <_\text{lex} \langle \lineIndex(\aLine(\aTrip')), \arrivalTime(\aTrip'[0]), j \rangle$, where $<_\text{lex}$ denotes the lexicographical order.} 
	Note that this ordering is consistent with the partial order $\prec$ on stop events, i.e., $\aLine(\aTrip)=\aLine(\aTrip')$ and $\aTrip[i] \prec \aTrip'[j]$ implies $\eventIndex(\aTrip[i]) < \eventIndex(\aTrip'[j])$.
	
	For a journey $\aJourney=\langle \footpath_0, \dots, \tripSegment{\aTrip}{i}{j}, (\stopA, \stopB)\rangle$, the~\emph{tiebreakers} are given by
	\[
	\tiebreakers(\aJourney):=\begin{cases}
		\langle\eventIndex(\aTrip[i]),\mathmakebox[0pt][l]{\infty}\hphantom{\eventIndex(\aTrip[j])}\rangle&\text{if }\stopA = \stopB,\\
		\langle\mathmakebox[0pt][l]{\infty}\hphantom{\eventIndex(\aTrip[i])},\eventIndex(\aTrip[j])\rangle&\text{otherwise.}\\
	\end{cases}
	\]
	The~\emph{local tiebreaking sequence} $	\localTiebreakingSequence(\aJourney):=\langle\arrivalTime(\aJourney)\rangle\circ\tiebreakers(\aJourney)$ consists of the arrival time followed by the tiebreakers.
	The global tiebreaking sequence of a journey $\aJourney$ with stop sequence $\stops(\aJourney)=\langle\aStop_1,\dots,\aStop_k\rangle$ concatenates the local tiebreaking sequences in reverse order:
	\[\tiebreakingSequence(\aJourney):=\localTiebreakingSequence(\subjourney{\aStop_1}{\aStop_k})\circ\dots\circ\localTiebreakingSequence(\subjourney{\aStop_1}{\aStop_2}).\]

    \changed{Sequences are compared via the lexicographical order $<_\text{lex}$.}
    If one sequence is shorter than the other, it is padded by adding entries with value $-\infty$ at the end.
	A journey $\aJourney$ is called~\emph{canonical} for a query $\query$ if it is Pareto-optimal under the \changed{cost system $\costsystem_\tiebreakingSequence=((\tiebreakingSequence,<_\text{lex}),(\absoluteVal{\cdot},\leq))$.}
	Because the tiebreaking sequences of all feasible journeys for $\query$ are unique and totally ordered, no two journeys are equivalent under $\costsystem_\tiebreakingSequence$.
	Accordingly, the only representative set for the query is the \emph{canonical representative set}, which consists solely of canonical journeys.
    \changed{A crucial property is that the set of all canonical journeys for all possible queries is closed under subjourney decomposition:}
    \begin{lemma}\label{th:canonical-subjourney}
        \changed{Let $\aJourney$ be a canonical journey for a query $\query$ and let $\subjourney{\stopA}{\stopB}$ be a subjourney of $\aJourney$ for $\stopA,\stopB\in\stops(\aJourney)$.
        Then $\subjourney{\stopA}{\stopB}$ is canonical for the query $(\stopA,\stopB,\departureTime(\subjourney{\stopA}{\stopB}))$.}
    \end{lemma}
    \begin{proof}
        \changed{This follows from the fact that adding a common prefix or suffix to two journeys does not change the relative order of their tiebreaking sequences.}
    \end{proof}
    
	Hence, the set of transfers that occur in at least one canonical candidate is canonicity-preserving~(CP).
	Furthermore, Lemma~\ref{th:canonical_prefix_optimal} shows that canonical journeys are prefix-optimal and therefore also exit-optimal.
	Hence, the set is EOP and can be used for one-to-all-TB by Theorem~\ref{th:one-to-all:correct}
	\begin{lemma}\label{th:canonical_prefix_optimal}
		Every canonical journey $\aJourney$ for a query $\query$ is prefix-optimal for $\query$.
	\end{lemma}
	\begin{proof}
		Let $\stops(\aJourney)=\langle\aStop_1,\dots,\aStop_k\rangle$ be the stop sequence of $\aJourney$.
		For each proper prefix $\subjourney{\aStop_1}{\aStop_n}$, the local tiebreaking sequence $\localTiebreakingSequence(\subjourney{\aStop_1}{\aStop_n})$ begins with $\arrivalTime(\subjourney{\aStop_1}{\aStop_n})$, so $\subjourney{\aStop_1}{\aStop_n}$ is Pareto-optimal under $\costsystem$.
		Let $\tripSegment{\aTrip}{i}{j}$ denote the last trip segment of $\subjourney{\aStop_1}{\aStop_n}$.
		Assume that the partial prefix of $\aJourney$ that ends with $\aTrip[i]$ is not Pareto-optimal under $\costsystem_p$.
		Because $\subjourney{\aStop_1}{\aStop_n}$ is Pareto-optimal, there must be a journey $\subjourneyAlt{\aStop_1}{\aStop_n}$ with $\cost(\subjourneyAlt{\aStop_1}{\aStop_n})=\cost(\subjourney{\aStop_1}{\aStop_n})$ that ends with a trip segment $\tripSegment{\aTrip'}{i'}{j}$ such that $\aTrip'[i']\prec\aTrip[i]$.
		Because trips of the same line do not overtake each other, this requires that $\aTrip'=\aTrip$ and $i'<i$ and therefore $\eventIndex(\aTrip'[i'])<\eventIndex(\aTrip[i])$.
		However, then it follows that~\changed{$\tiebreakers(\subjourneyAlt{\aStop_1}{\aStop_n})<_\text{lex}\tiebreakers(\subjourney{\aStop_1}{\aStop_n})$}, which contradicts the fact that $\aJourney$ is canonical.
	\end{proof}
	
	Canonical representative sets can be computed with a slight modification of RAPTOR~\cite{Del15b} called~\emph{canonical RAPTOR}.
	Like~TB, RAPTOR operates in rounds.
	Each round first explores the outgoing lines and then the outgoing footpaths of previously reached stops.
	To ensure that canonical journeys are found, canonical RAPTOR makes two modifications:
	The collected lines are sorted according to $\lineIndex$ before they are explored.
	For each stop $\aStop$ reached directly via a trip, the algorithm maintains the stop event $\stopEvent(\aStop)$ at which the trip was exited.
	The collected stops are then sorted according to $\eventIndex(\stopEvent(\aStop))$ before their outgoing footpaths are explored.

    \subparagraph*{\changed{Repairing Self-Pruning.}}
    Enumerating all canonical candidates requires solving the all-to-all full-range profile problem, restricted to journeys with at most two trips.
	\transultra does this with rRAPTOR, an extension of RAPTOR for profile queries.
	In the same manner as Profile-TB, rRAPTOR performs a canonical RAPTOR run for each possible departure time in decreasing order.
	Because candidates have two trips, each run can be stopped after two rounds.
	As we saw in Section~\ref{sec:transferset:issues}, the self-pruning rule cannot be applied directly because it conflicts with exit optimality.

    \changed{This issue can be solved by exploiting the following insight:
    If a journey $\aJourney$ is canonical for a query $(\sourceStop,\targetStop,\departureTime)$, then it is also canonical for the query $(\sourceStop,\targetStop,\departureTime(\aJourney))$.
    If the objective is to enumerate all canonical journeys (or candidates), then we only need to ensure that $\aJourney$ is chosen as the representative during the run for $\departureTime(\aJourney)$.
    In subsequent runs (for earlier departure times), $\aJourney$ may still be canonical, but because it is already included in the algorithm's output, it is safe to replace it with a different journey.}

    \changed{This insight leads to a variant of rRAPTOR called \emph{canonical rRAPTOR}.}
	Each RAPTOR run is done with canonical RAPTOR.
	Additionally, the self-pruning rule is relaxed to ensure that \changed{exit-optimal journeys are not discarded if they are dominated by an equivalent journey (under $\costsystem$) with a later departure time}.
	The run in which a journey $\aJourney$ is found is the one for $\departureTime(\aJourney)$.
	Let $\aJourney'$ the journey with the earliest arrival time among those with at most $\absoluteVal{\aJourney}$ trips that were previously found at $\targetStop(\aJourney)$.
	Usually, RAPTOR discards $\aJourney$ if $\aJourney' \preceq \aJourney$ holds under $\costsystem$.
	Otherwise, $\aJourney$ replaces $\aJourney'$.
	\changed{Canonical rRAPTOR instead discards $\aJourney$ if one of the two~\emph{ULTRA conditions} holds under $\costsystem$:}
	\begin{description}
		\item[(U1)] $\aJourney' \preceq \aJourney$ and $\departureTime\left(\aJourney'\right) = \departureTime\left(\aJourney\right)$.
		\item[(U2)] $\aJourney' \prec \aJourney$.
	\end{description}
	\changed{With these conditions, we show canonical rRAPTOR returns every canonical journey:}
	\begin{lemma}\label{th:departure_canonical_mr}
		Let $\aJourney$ be a \changed{canonical} journey for a query $\query=(\sourceStop,\targetStop,\departureTime)$.
		After the run for~\changed{$\departureTime(\aJourney)$ in the canonical rRAPTOR search} from $\sourceStop$, the representative for $\cost(\aJourney)$ at $\targetStop$ is $\aJourney$.
	\end{lemma}
	\begin{proof}
		Note that the ULTRA conditions do not affect the computed cost vectors, only the parent pointers (and thereby the representative).
		Hence, it follows from the correctness of canonical RAPTOR and rRAPTOR that \changed{canonical rRAPTOR} returns the \changed{Pareto front}, which includes $\cost(\aJourney)$.
		Let $\aJourney'$ be the corresponding representative found by \changed{canonical rRAPTOR}.
		We show that $\aJourney'=\aJourney$.
		If $\departureTime(\aJourney')=\departureTime(\aJourney)$, then it follows from the correctness of canonical RAPTOR that $\aJourney'$ is the journey with the smallest tiebreaking sequence among those in $\journeys(\query)$ with cost vector $\cost(\aJourney)$ and departure time $\departureTime(\aJourney)$, which is $\aJourney$.
		Assume therefore that $\departureTime(\aJourney')>\departureTime(\aJourney)$.
		We prove the claim by induction over the proper prefixes of $\aJourney$.
		Because every proper prefix of $\aJourney$ is \changed{canonical} by~\changed{\Cref{th:canonical_prefix_optimal}}, the algorithm returns it as a representative by the induction hypothesis.
		We therefore know that the run for $\departureTime(\aJourney)$ finds $\aJourney$ and compares it to $\aJourney'$.
		Condition (U1) is not fulfilled by our assumption that $\departureTime(\aJourney')>\departureTime(\aJourney)$, whereas condition (U2) is not fulfilled because $\aJourney$ is Pareto-optimal.
		Hence, $\aJourney$ is not discarded and replaces $\aJourney'$, a contradiction.
	\end{proof}
	
	\subparagraph*{\transultra.}
	\transultra performs a \changed{canonical} rRAPTOR search from each possible source stop, restricted to the first two rounds.
	For each candidate in this set, its transfer is added to the transfer set $\transfers$.
	We show that $\transfers$ is \changed{canonicity-preserving}.
	
	\begin{theorem}\label{th:ultra:pdop}
		The transfer set $\transfers$ generated by \transultra is \changed{canonicity-preserving}.
	\end{theorem}
	\begin{proof}
        Consider a \changed{canonical} representative $\aJourney=\left<\tripSegment{\aTrip_1}{i_1}{j_1},\dots,\tripSegment{\aTrip_k}{i_k}{j_k}\right>$ for some query $\query$ and, for each $1 \leq n < k$, the candidate subjourney $\aJourney_n=\left<\tripSegment{\aTrip_n}{i_n}{j_n},\tripSegment{\aTrip_{n+1}}{i_{n+1}}{j_{n+1}}\right>$.
		\changed{By~\Cref{th:canonical-subjourney}, $\aJourney_n$ is canonical for the query $(\sourceStop(\aJourney_n),\targetStop(\aJourney_n),\departureTime(\aJourney_n))$}.
		Hence, the transfer $\transfer{\aTrip_n[j_n]}{\aTrip_{n+1}[i_{n+1}]}$ is included in $\transfers$ by Lemma~\ref{th:departure_canonical_mr}.
	\end{proof}
	
	Note that the proof of Theorem~\ref{th:ultra:pdop} only requires a weaker version of Lemma~\ref{th:departure_canonical_mr}: it is sufficient to find and unpack the \changed{canonical} representatives that are candidates.
	ULTRA exploits this with several performance optimizations that prevent irrelevant non-candidate journeys from being explored~\cite{Bau23}.
	These optimizations carry over to \transultra; see Appendix~\ref{app:optimizations}.

	\subsection{\changed{Canonical Profile-TB}}
	We now return to the task of precomputing auxiliary data to speed up TB.
	The precomputation scheme is as follows:
    Using the \transultra transfer set, the all-to-all full-range profile problem is solved by running one-to-all Profile-TB from every source stop.
	For every run and every target stop, the representative journeys are unpacked and the auxiliary data is computed from them.
	It follows from Theorems~\ref{th:one-to-all:correct} and~\ref{th:ultra:pdop} that \changed{each profile search computes the Pareto front}, but as outlined in Section~\ref{sec:transferset:issues}, the representatives are not necessarily compatible with the pruning rules of the query algorithm.
    Therefore, we perform the profile searches with \changed{\emph{canonical Profile-TB}}, a variant of one-to-all Profile-TB that computes all \changed{canonical} representatives when it is used with a \changed{CP} transfer set (such as the one computed by \transultra).
	
	Whereas RAPTOR always finds prefix-optimal representatives, TB in its original form only guarantees exit optimality via the local pruning rule.
	We add a new pruning rule to ensure that \changed{Pareto-optimal} footpath prefixes are selected.
	Furthermore, although the reached index data structure ensures that the algorithm finds optimal partial prefixes, they are not necessarily selected during journey unpacking.
	To change this, we redesign the parent pointers used for journey unpacking.
	We also adjust the order in which trip segments, footpaths and transfers are scanned within each round to match the order given by the tiebreakers.
	This ensures that each run finds the \changed{canonical} representative first and therefore discards all others.
	Finally, we relax self-pruning by incorporating the ULTRA conditions.
	Pseudocode for the modified parts is given in Algorithm~\ref{alg:flashtb:preprocessing:scan}.
	
	\subparagraph*{Relaxing-Self Pruning.}
	Like RAPTOR, one-to-all TB represents found journeys implicitly via the earliest arrival time $\arrivalTime(\aStop,n)$ for each stop $\aStop$ and round $n$.
	We incorporate the ULTRA conditions in the procedure $\AddArrival$, which updates this value.
	Unlike RAPTOR, TB also represents partial journeys implicitly via the reached index $\reachedIndex_n(\aTrip)$ for each trip $\aTrip$ and round $n$.
	When this is updated in $\Enqueue$, a version of the ULTRA conditions also needs to be incorporated.
	
	We represent the run of a partial journey implicitly by maintaining a secondary reached index $\runReachedIndex(\cdot)$, which is reset after every run.
	Consider a call to $\Enqueue$ for a stop event $\aTrip[j]$ in round $n$.
	Let $\aJourney$ denote the corresponding partial journey up to $\aTrip[j]$.
	Condition~(U1) becomes condition~(E1), which is evaluated in line~\ref{alg:flashtb:preprocessing:enqueue:e1}: if $\runReachedIndex\left(\aTrip\right) \leq j$, then the algorithm has already found a partial journey $\aJourney'$ with $\aJourney'\preceq\aJourney$ and $\departureTime(\aJourney')=\departureTime(\aJourney)$.
	Condition (U2) can be tested by evaluating the reached index $\reachedIndex_n\left(\aTrip\right)$.
	To check whether the algorithm has already found a partial journey $\aJourney'\prec\aJourney$, the algorithm evaluates three sub-conditions:
	\begin{description}
		\item[(E2a)] If $\aJourney'$ ends with a stop event $\aTrip[\ell]$ such that $\ell < j$, then $\reachedIndex_n\left(\aTrip\right) < j$ must hold (see line~\ref{alg:flashtb:preprocessing:enqueue:e2a}).
		\item[(E2b)] If $\absoluteVal{\aJourney'}<\absoluteVal{\aJourney}$, then $\reachedIndex_{n-1}\left(\aTrip\right) \leq j$ must hold (see line~\ref{alg:flashtb:preprocessing:enqueue:e2b}).
		\item[(E2c)] If $\aJourney'$ ends with a stop event $\aTrip'[\ell]$ such that $\aTrip' \prec \aTrip$ and $\ell \leq j$, then $\reachedIndex_n(\aTrip') \leq j$ must hold.
		Then it follows that $\pred{\aTrip}\neq\bot$ and $\reachedIndex_n(\pred{\aTrip}) \leq j$ (see line~\ref{alg:flashtb:preprocessing:enqueue:e2c}).
	\end{description}
	
	Consider now an $\AddArrival$ call for a stop $\aStop$ with arrival time $\arrivalTime$ in round $n$ of the run with departure time $\departureTime$.
	Besides the earliest found arrival time $\arrivalTime(\aStop, n)$ at $\aStop$ in round $n$, the algorithm maintains the departure time of the run in which $\arrivalTime(\aStop, n)$ was last updated, which we denote by $\departureTime(\aStop, n)$.
	The new arrival is discarded if one of the following conditions holds:
	\begin{description}
		\item[(T1)] $\arrivalTime(\aStop,n) = \arrivalTime$ and $\departureTime(\aStop,n) = \departureTime$,
		\item[(T2a)] $\arrivalTime(\aStop,n) < \arrivalTime$,
		\item[(T2b)] $\arrivalTime(\aStop, n-1) \leq \arrivalTime$.
	\end{description}
	Otherwise, $\arrivalTime(\aStop, n)$ is set to $\arrivalTime$ and $\departureTime(\aStop, n)$ to $\departureTime$.
	
	\subparagraph*{Reordering the Trip Scanning Phase.}
	At the start of the trip scanning phase of round $n$, the trip segments $\tripSegment{\aTrip}{j}{k}$ in the queue $\aQueue_n$ are sorted in ascending order of $\eventIndex(\aTrip[j])$ (see line~\ref{alg:flashtb:preprocessing:scan:sort}).
    This ensures that all stop events are scanned in ascending order according to $\eventIndex$.
	Normally, one-to-all TB loops over the trip segments twice: once to update the earliest arrival times (done here via $\AddArrival$) and then again to scan outgoing transfers.
	\changed{Canonical} Profile-TB splits each of these loops into two parts, resulting in four loops in total.
    The first loop only calls $\AddArrival$ for the stops at which the trips are exited.
    The second loop explores the outgoing footpaths and calls $\AddArrival$ for the stops reached in this manner.
    This ensures that journeys with empty final footpaths, which have smaller tiebreakers than those with non-empty ones, are found first.
	To achieve the same exploration order for the transfers, the transfer set $\transfers$ is split into two sets: the set $\directTransfers$ contains the transfers between stop events at the same stop, whereas $\footpathTransfers$ contains those that involve a non-empty footpath.
	The third loop explores the transfers in $\directTransfers$, whereas the fourth loop explores those in $\footpathTransfers$.
	Along with the sorting of the trip segments, these changes ensure that \changed{within each run}, the first found representative for each Pareto-optimal cost vector is the \changed{canonical} one.
	
	\subparagraph*{Local Pruning.}
	The loops that explore outgoing transfers retain the local pruning rule from one-to-all TB (see lines~\ref{alg:flashtb:preprocessing:scan:localPruning1} and~\ref{alg:flashtb:preprocessing:scan:localPruning2}), which ensures that journeys with non-optimal exit prefixes are discarded.
	A similar pruning rule is added in line~\ref{alg:flashtb:preprocessing:scan:footpathPruning}.
	When exploring a transfer $(\aTrip[i],\aTrip'[i'])$ that involves a non-empty footpath, the algorithm needs to ensure that the resulting journey ending at $\aStop(\aTrip'[i'])$ is optimal.
	This is done by calculating its arrival time $\arrivalTime$ in line~\ref{alg:flashtb:preprocessing:scan:footpathTime} and comparing it to the earliest found arrival time $\arrivalTime(\aStop(\aTrip'[i']),n)$.
	If it is greater, then the journey is not \changed{Pareto-}optimal, so $\Enqueue$ is not called.
	Note that this pruning rule is not necessary when exploring the transfers $\directTransfers$ that do not involve a footpath.
	Also note that the arrival time $\arrivalTime$ calculated in line~\ref{alg:flashtb:preprocessing:scan:footpathTime} depends only on the transfer, so it can be calculated and stored in advance for each transfer.
	
	\subparagraph*{Journey Unpacking.}
    Another modification is required to ensure that the unpacked journeys enter each trip at the earliest possible stop.
	When $\Enqueue$ is called for a stop event $\aTrip[i]$, a trip segment $\tripSegment{\aTrip}{j}{k}$ with $j > k$ may have already been enqueued earlier.
	In this case, the algorithm enqueues the segment $\tripSegment{\aTrip}{i}{j-1}$ as well, so the queue contains two consecutive trip segments of the same trip.
	This means that for every journey that exits $\aTrip$ at some stop event $\aTrip[\ell]$ with $j < \ell \leq k$, there is a journey that uses the trip segment $\tripSegment{\aTrip}{i}{\ell}$ and one that uses $\tripSegment{\aTrip}{j}{\ell}$.
	The former must be chosen because the partial prefix ending in $\aTrip[i]$ strictly dominates the one ending in $\aTrip[j]$.
	However, TB will normally choose the latter because it maintains a parent pointer per enqueued trip segment.

    \changed{Canonical} Profile-TB instead maintains a parent pointer $\parent(\aTrip,n)$ for each trip $\aTrip$ and round $n$.
	If $\aTrip[i]$ is the earliest reachable stop event of $\aTrip[i]$ and $(\aStop,\aStop(\aTrip[i]))$ is the (possibly empty) footpath from which it was reached, then $\parent(\aTrip,n)$ stores $\aStop$.
	If $n=1$, this is the source stop; otherwise, it is the stop at which the previous trip was exited.
	If an $\Enqueue$ call for the stop event $\aTrip[j]$ is successful (\ie none of the ULTRA conditions hold), the parent pointer is updated because $\aTrip[j]$ is now the earliest reachable stop event.
	Additionally, a pointer $\parent(\aStop,n)$ is stored for each stop $\aStop$ and round $n$, which stores the last trip segment in the journey that was used to reach $\aStop$.
	Consider a call to $\AddArrival$ during the scan of a trip segment $\tripSegment{\aTrip}{j}{k}$ in line~\ref{alg:flashtb:preprocessing:scan:arrivalDirect} or~\ref{alg:flashtb:preprocessing:scan:arrivalFootpath}.
	The currently scanned stop event $\aTrip[i]$ is passed to $\AddArrival$ as the stop event at which the trip segment is exited.
	The stop event at which it is entered is not necessarily $\aTrip[j]$, but rather the stop event $\aTrip[\ell]$ at which the earliest enqueued trip segment of $\aTrip$ starts.
	At this point during the execution of the algorithm, $\ell$ is given by $\runReachedIndex(\aTrip)-1$.
	
	After each run with departure time $\departureTime$, journey unpacking is performed for each combination of stop $\stopA$ and round $n$ for which $\departureTime(\stopA,n)=\departureTime$.
	To unpack the journey, the algorithm retrieves the trip segment $\tripSegment{\aTrip}{j}{k}$ stored in $\parent(\stopA,n)$.
	Then it looks up the stop $\stopB$ stored in $\parent(\aTrip,n)$ and repeats the process recursively with $\stopB$ and $n-1$.
	This is repeated until $n$ reaches $0$.
	To avoid looping over all stops and rounds after every run, all pairs of $\stopA$ and round $n$ for which $\departureTime(\stopA,n)$ is updated are marked during the run.
	
	\begin{figure}
		\begin{minipage}{\textwidth}
			\begin{algorithm}[H]
				\caption{Trip scanning operation of DC-Profile-TB.}
				\label{alg:flashtb:preprocessing:scan}
				\myproc{\Scan{$n, \departureTime$}}{
					$\aQueue_{n+1}\leftarrow{}\emptyset$\;
					Sort trip segments $\tripSegment{\aTrip}{j}{k}$ in $\aQueue_n$ by $\eventIndex(\aTrip[j])$\label{alg:flashtb:preprocessing:scan:sort}\;
					\ForEach{$\tripSegment{\aTrip}{j}{k}\in\aQueue_n$}{
						\For{$i$ from $j$ to $k$}{
							$\AddArrival(\aStop(\aTrip[i]), n, \arrivalTime(\aTrip[i]), \departureTime, \aTrip, i)$\label{alg:flashtb:preprocessing:scan:arrivalDirect}\;
						}
					}
					\ForEach{$\tripSegment{\aTrip}{j}{k}\in\aQueue_n$}{
						\For{$i$ from $j$ to $k$}{
							\ForEach{$\aStop \in \stops \setminus \{\aStop(\aTrip[i])\}$ with $\transfertime{\aStop(\aTrip[i])}{\aStop}<\infty$}{
								$\arrivalTime \leftarrow \arrivalTime(\aTrip[i])+\transfertime{\aStop(\aTrip[i])}{\aStop}$\;
								$\AddArrival(\aStop, n, \arrivalTime, \departureTime, \aTrip, i)$\label{alg:flashtb:preprocessing:scan:arrivalFootpath}\;
							}
						}
					}
					\ForEach{$\tripSegment{\aTrip}{j}{k}\in\aQueue_n$}{
						\For{$i$ from $j$ to $k$}{
							\lIfComment{local pruning}{$\arrivalTime\left(\aTrip[i]\right)>\arrivalTime(\aStop(\aTrip[i]), n)$\label{alg:flashtb:preprocessing:scan:localPruning1}}{\Continue}
							\ForEachComment{scan transfers}{$\transfer{\aTrip[i]}{\aTrip'[i']}\in\directTransfers$}{
								$\Enqueue\left(\aTrip', i'+1, n+1, \aStop(\aTrip[i])\right)$
							}
						}
					}
					\ForEach{$\tripSegment{\aTrip}{j}{k}\in\aQueue_n$}{
						\For{$i$ from $j$ to $k$}{
							\lIfComment{local pruning}{$\arrivalTime\left(\aTrip[i]\right)>\arrivalTime(\aStop(\aTrip[i]), n)$\label{alg:flashtb:preprocessing:scan:localPruning2}}{\Continue}
							\ForEachComment{scan transfers}{$\transfer{\aTrip[i]}{\aTrip'[i']}\in\footpathTransfers$}{
								$\arrivalTime \leftarrow \arrivalTime\left(\aTrip[i]\right)+\transfertime{\aStop(\aTrip[i])}{\aStop(\aTrip'[i'])}$\label{alg:flashtb:preprocessing:scan:footpathTime}\;
								\lIfComment{local pruning}{$\arrivalTime >\arrivalTime(\aStop(\aTrip'[i']), n)$\label{alg:flashtb:preprocessing:scan:footpathPruning}}{\Continue}
								$\Enqueue\left(\aTrip', i'+1, n+1, \aStop(\aTrip[i])\right)$
							}
						}
					}
				}
				
				\myproc{$\AddArrival(\aStop, n, \arrivalTime, \departureTime, \aTrip, i)$}{
					\lIfComment{(T1)\phantom{a}}{$\arrivalTime(\aStop, n) = \arrivalTime$ and $\departureTime(\aStop, n) =\departureTime$\label{alg:flashtb:preprocessing:scan:t1}}{\Return}
					\lIfComment{(T2a)}{$\arrivalTime(\aStop, n) < \arrivalTime$\label{alg:flashtb:preprocessing:scan:t2a}}{\Return}
					\lIfComment{(T2b)}{$\arrivalTime(\aStop, n-1) \leq \arrivalTime$\label{alg:flashtb:preprocessing:scan:t2b}}{\Return}
					\ForEach{$m \geq n$}{
						\lIf{$\arrivalTime(\aStop,m) \leq \arrivalTime$}{\Break}
						$\arrivalTime(\aStop, m) \leftarrow \arrivalTime$\;
						$\departureTime(\aStop, m) \leftarrow \departureTime$\;
					}
					$\parent(\aStop, n) \leftarrow \aTrip[\runReachedIndex(\aTrip)-1, i]$\;
				}
	
				\myproc{\Enqueue{$\aTrip, j, n, \aStop$}}{
					\lIfComment{(E1)\phantom{a}}{$\runReachedIndex(\aTrip)\leq{}j$\label{alg:flashtb:preprocessing:enqueue:e1}}{\Return}
					\lIfComment{(E2a)}{$\reachedIndex_{n}(\aTrip) < j$\label{alg:flashtb:preprocessing:enqueue:e2a}}{\Return}
					\lIfComment{(E2b)}{$n > 1 \wedge \reachedIndex_{n-1}(\aTrip)\leq{}j$\label{alg:flashtb:preprocessing:enqueue:e2b}}{\Return}
					\lIfComment{(E2c)}{$\pred{\aTrip} \neq \bot \wedge \reachedIndex_n(\pred{\aTrip})\leq{}j$\label{alg:flashtb:preprocessing:enqueue:e2c}}{\Return}
					$\aQueue_n\leftarrow\aQueue_n\cup\{\tripSegment{\aTrip}{j}{\runReachedIndex(\aTrip)-1}\}$\;
					\ForEach{$\aTrip'\succeq\aTrip$}{
						$\runReachedIndex(\aTrip')\leftarrow{}\min\left\{j, \runReachedIndex(\aTrip')\right\}$\;
						\lForEach{$m \geq n$}{$\reachedIndex_{m}(\aTrip')\leftarrow{}\min\left\{j, \reachedIndex_{m}(\aTrip')\right\}$
						}
					}
					$\parent(\aTrip',n) \leftarrow \aStop$\;
				}
			\end{algorithm}
            \vspace*{-20pt}
		\end{minipage}
	\end{figure}
	
	\subparagraph*{Proof of Correctness.}
	We now show that \changed{canonical} Profile-TB with a \changed{CP} transfer set finds and unpacks all \changed{canonical} journeys.

	\begin{theorem}
		\label{th:flashtb:profile}
		Let $\query=(\sourceStop,\targetStop,\departureTime)$ be a query, $\transfers$ a \changed{CP} transfer set and $\aJourney=\left< \footpath_0, \tripSegment{\aTrip_1}{i_1}{j_1}, \dots, \tripSegment{\aTrip_k}{i_k}{j_k}, \footpath_{k+1} \right> \in \journeys(\query,\transfers)$ a \changed{canonical} journey for $\query$.
		In the \changed{canonical} Profile-TB search from $\sourceStop$, the run for $\departureTime(\aJourney)$ unpacks $\aJourney$.
	\end{theorem}
	\begin{proof}
		To simplify the notation, we define $\inStop_n:=\aStop(\aTrip_n[i_n])$ and $\outStop_n:=\aStop(\aTrip_n[j_n])$ for $1 \leq n \leq k$.
		Additionally, we define $\outStop_0:=\sourceStop$ and $\inStop_{k+1}:=\targetStop$.
		We show that after the run for $\departureTime(\aJourney)$ has finished, the following holds for all $1 \leq n \leq k$:
		\begin{itemize}
			\item The run has called $\Enqueue(\aTrip_n,i_n+1,n,\outStop_{n-1})$ and it was the first call of the form $\Enqueue(\aTrip_n,i_n+1,n,\cdot)$ in this run.
			\item $\departureTime(\outStop_n,n)=\departureTime(\inStop_{n+1},n)=\departureTime(\aJourney)$,
			\item $\parent(\outStop_n,n)=\parent(\inStop_{n+1},n)=\tripSegment{\aTrip_n}{i_n}{j_n}$,
			\item $\parent(\aTrip_n,n)=\outStop_{n-1}$.
		\end{itemize}
		Then it follows that the algorithm unpacks $\aJourney$.
		We prove these claims by induction over $n$.
		In round $0$, the algorithm explores the outgoing footpaths of $\sourceStop$, so all $\Enqueue$ calls have the parent stop $\outStop_0=\sourceStop$.
		When the initial footpath $\footpath_0$ is explored, the call $\Enqueue(\aTrip_1,i_1+1,1,\outStop_0)$ is made.
		
		Assume now that the claims hold for $n-1$.
		In particular, $\Enqueue(\aTrip_n,i_n+1,n,\outStop_{n-1})$ was called in round $n-1$ and it was the first call of the form $\Enqueue(\aTrip_n,i_n+1,n,\cdot)$ in this run.
		Consider what happens during this call.
		Let $\aJourney_p^n$ denote the partial prefix of $\aJourney$ that ends with $\aTrip_n[i_n]$.
		Because $\aJourney_p^n$ is Pareto-optimal, no call $\Enqueue(\aTrip',i'+1,n',\cdot)$ with $\aTrip'[i'] \prec \aTrip_n[i_n]$ and $n' \leq n$ is made by the algorithm.
		This implies that $\runReachedIndex(\aTrip_n) > i_n + 1$ holds at the beginning of the current call, so pruning condition~(E1) is not fulfilled.
		Conditions~(E2) are not fulfilled because $\aJourney_p^n$ is Pareto-optimal.
		It follows that $\parent(\aTrip_n,n)$ is set to $\outStop_{n-1}$.
        During any subsequent call of the form $\Enqueue(\aTrip_n,i_n,n,\cdot)$, condition~(E1) is fulfilled, so $\parent(\aTrip_n,n)$ is not overwritten.
		
		Because $\subjourney{\sourceStop}{\outStop_n}$ is Pareto-optimal, $\runReachedIndex(\aTrip_n) > j_n$ must hold before the start of round $n-1$.
		However, $\runReachedIndex(\aTrip_n) = i_n + 1 \leq j_n$ holds after the call $\Enqueue(\aTrip_n,i_n+1,n,\outStop_{n-1})$, so a trip segment containing $\aTrip_n[j_n]$ is enqueued and then scanned in round $n$.
		Let $\aStop \in \{ \outStop_n, \inStop_{n+1} \}$.
		When $\aTrip_n[j_n]$ is scanned, \changed{$\AddArrival(\aStop,n,\arrivalTime(\subjourney{\sourceStop}{\aStop}),\departureTime(\aJourney),\aTrip_n,j_n)$} is called.
        The pruning conditions~(T2) are not fulfilled in this call because $\subjourney{\sourceStop}{\aStop}$ is Pareto-optimal.
        We show that~(T1) is not fulfilled either.
		Assume that a previous call~\changed{$\AddArrival(\aStop,n,\arrivalTime(\subjourney{\sourceStop}{\aStop}),\departureTime(\aJourney),\aTrip',j')$} was made during this run.
		Then there is a $\sourceStop$-$\aStop$-journey $\aJourney'$ with departure time~\changed{$\departureTime(\aJourney)$} and arrival time $\arrivalTime(\aJourney')=\arrivalTime(\subjourney{\sourceStop}{\aStop})$ that ends with a trip segment $\tripSegment{\aTrip'}{i'}{j'}$.
		If $\aStop\neq\outStop_n$, then the order in which the enqueued trip segments are scanned implies that~\changed{$\aStop(\aTrip'[j'])=\aStop$} or $\eventIndex(\aTrip'[j'])<\eventIndex(\aTrip_n[j_n])$.
		If $\aStop=\outStop_n$, then both must hold.
		Both cases imply~\changed{$\tiebreakingSequence(\aJourney') <_\text{lex} \tiebreakingSequence(\subjourney{\sourceStop}{\aStop})$}, which contradicts the \changed{canonicity} of $\aJourney$ \changed{by~\Cref{th:canonical-subjourney}}.
		Hence, $\departureTime(\aStop,n)$ is set to $\departureTime(\aJourney)$ and will retain this value for the remainder of the run.
		Similarly, $\arrivalTime(\aStop,n)$ is set to $\arrivalTime(\aJourney,\aStop)$ and $\parent(\aStop,n)$ is set to $\tripSegment{\aTrip_n}{\runReachedIndex(\aTrip_n)-1}{j_n}=\tripSegment{\aTrip_n}{i_n}{j_n}$.
		No subsequent call of the form $\AddArrival(\aStop,n,\arrivalTime,\departureTime(\aJourney),\cdot,\cdot)$ is able to overwrite the parent pointer:
        Because $\subjourney{\sourceStop}{\aStop}$ is Pareto-optimal, $\arrivalTime\geq\arrivalTime(\aJourney,\aStop)$ must hold.
		If $\arrivalTime>\arrivalTime(\aJourney,\aStop)$ holds, then condition~(T2a) applies.
		If $\arrivalTime=\arrivalTime(\aJourney,\aStop)$ holds, then condition~(T1) applies.
		
		During the scan of $\aTrip_n[j_n]$ in round $n<k$, the outgoing transfer $\transfer{\aTrip_n[j_n]}{\aTrip_{n+1}[i_{n+1}]}\in\transfers$ is scanned.
		Because $\subjourney{\sourceStop}{\outStop_n}$ and $\subjourney{\sourceStop}{\inStop_{n+1}}$ are Pareto-optimal, the local pruning conditions are not fulfilled, so $\Enqueue(\aTrip_{n+1},i_{n+1}+1,n+1,\outStop_n)$ is called.
		If this is not the first call of the form $\Enqueue(\aTrip_{n+1},i_{n+1}+1,n+1,\cdot)$ in this run, then there is a $\sourceStop$-$\inStop_{n+1}$-journey $\aJourney'$ with departure time~\changed{$\departureTime(\aJourney)$} and final trip segment $\tripSegment{\aTrip'}{i'}{j'}$ such that $\transfer{\aTrip'[j']}{\aTrip_{n+1}[i_{n+1}]}$ is scanned before $\transfer{\aTrip_n[j_n]}{\aTrip_{n+1}[i_{n+1}]}$.
		Because the transfer is only scanned if the local pruning condition is not fulfilled, $\arrivalTime(\aJourney')\geq\arrivalTime(\subjourney{\sourceStop}{\inStop_{n+1}})$ must hold.
		If $\outStop_n=\inStop_{n+1}$, this implies $\aTrip'[j']\prec\aTrip_n[j_n]$.
		Otherwise, it implies $\aStop(\aTrip'[j'])=\inStop_{n+1}$ or $\aTrip'[j']\prec\aTrip_n[j_n]$.
		In either case, \changed{$\tiebreakingSequence(\aJourney') <_\text{lex} \tiebreakingSequence(\subjourney{\sourceStop}{\inStop_{n+1}})$} follows, which contradicts the \changed{canonicity} of $\aJourney$ \changed{by~\Cref{th:canonical-subjourney}}.
	\end{proof}
	
	\section{FLASH-TB}
	\label{sec:flashtb}
	
	We now present our new algorithm FLASH-TB (TB with FLAgged SHortcuts), which applies the main idea of Arc-Flags to TB.
	
	\subsection{Overview}
	\label{sec:flashtb:basic}
	The first preprocessing step partitions the set $\stops$ of stops into $\numCells$ cells, which yields a partition function $\partition:\stops \to \left\{1, \dots, \numCells\right\}$.
	To represent the topology of the network without its time dependency, we define the~\emph{layout graph} $\layoutGraph$.
	The set of \emph{links} between a pair $\stopA,\stopB$ of stops is given by
	\[
	\links(\stopA,\stopB)\coloneqq\bigl\{ \tripSegment{\aTrip}{i}{i+1} \mid \aTrip\in\trips, \aStop(\aTrip[i])=\stopA, \aStop(\aTrip[i+1])=\stopB \bigr\} \cup \bigl\{ (\stopA,\stopB) \mid (\stopA,\stopB) \in \footpaths \bigr\}.
	\]
	Thus, a link is either a trip segment between two consecutive stops or a footpath.
	The layout graph (see Figure~\ref{fig:metisgraph}) condenses all links between the same pair of stops into a single edge.
	Formally, it is a graph $\layoutGraph=\left(\stops, \layoutEdges, \layoutEdgeWeight\right)$ with the edge set $\layoutEdges \coloneqq \left\{(\stopA,\stopB) \mid \links(\stopA,\stopB) \neq \emptyset \right\}$ and the edge weight function $\layoutEdgeWeight\left((\stopA,\stopB)\right)\coloneqq \absoluteVal{\links(\stopA,\stopB)}$.
	The stop partition $\partition$ is obtained by running a graph partitioning algorithm on $\layoutGraph$.
	Because the edges are weighted by the number of corresponding links, the partitioner is guided to avoid separating stops with many direct connections between them.
	
	\begin{figure}[tp!]
		\caption{
			\textit{Left:} An example network with stops as nodes, trips as colored edges and footpaths as gray edges.
			\textit{Right:} The corresponding layout graph with edges weighted by the number of links.
			Node groupings indicate a possible ${3}$-way partition of the graph.
		}
		\label{fig:metisgraph}
		\centering
		\input{fig/layoutGraphLeft}
		\hspace{2cm}
		\input{fig/layoutGraphRight}
	\end{figure}
	
	After the partition has been computed, the transfer set $\transfers$ is computed with~\transultra.
    Then, the algorithm computes a flag for each transfer $\aTransfer\in\transfers$ and cell $i$, which indicates whether $\aTransfer$ is required to reach any target stops in cell $i$.
	Formally, this yields a flags function $\aFlag : \transfers \times \left\{1, \dots, \numCells\right\} \to \left\{0, 1\right\}$ with the following property: for each query $\query=(\sourceStop,\targetStop,\departureTime)$, there is a representative set $\representativeSet$ such that $\aFlag(\aTransfer,\partition(\targetStop))=1$ for every transfer $\aTransfer=\transfer{\tripA[i]}{\tripB[j]}$ that occurs in a journey $\aJourney\in\representativeSet$.
	The flags are computed by solving the all-to-all full-range profile problem.
	As outlined in Section~\ref{sec:transferset}, this is done by running \changed{canonical} Profile-TB from every possible source stop.
	After each TB run, all newly found journeys are unpacked.
	For each found journey $\aJourney$ to a target stop $\targetStop$ and each transfer $\aTransfer$ in $\aJourney$, the flag $\aFlag(\aTransfer,\partition(\targetStop))$ is set to $1$.
	Transfers for which no flags are set to $1$ can be removed from $\transfers$.
	Afterwards, a query with target stop $\targetStop$ is answered by running the TB query algorithm but only scanning transfers that are flagged for the target cell $\partition(\targetStop)$.

    \begin{theorem}
		\label{th:flashtb:correct}
		FLASH-TB is correct.
	\end{theorem}
	\begin{proof}
		Let $\transfers$ be the \changed{(CP)} transfer set computed by \transultra and~\changed{$\aFlag(\cdot,\cdot)$} the computed flags function.
		For a cell $i$ of the stop partition, let $\transfers(i) = \{ \aTransfer \in \transfers \mid \aFlag(\aTransfer,i) = 1 \}$ be the set of flagged transfers.
		For a query $\query=(\sourceStop,\targetStop,\departureTime)$, the search space of FLASH-TB with the transfer set $\transfers$ is equal to that of TB with $\transfers(\partition(\targetStop))$.
		We show that $\transfers(\partition(\targetStop))$ is \changed{CP} for $\query$, which proves the claim by~\Cref{th:one-to-one}.
		Consider a \changed{canonical} journey $\aJourney$ for $\query$ and the \changed{canonical} Profile-TB search from $\sourceStop$ performed during the preprocessing.
		By Theorem~\ref{th:flashtb:profile}, the run for $\departureTime(\aJourney)$ unpacks $\aJourney$.
		For each transfer in $\aJourney$, the flag of the target cell $\partition(\targetStop)$ is set to $1$, so it is included in $\transfers(\partition(\targetStop))$.
	\end{proof}
	
	\subsection{Optimizations}
	\label{sec:flashtb:opt}
	
	\subparagraph*{Flag Representation.}
	For an edge $\anEdge$, we call the sequence $\langle\aFlag(\anEdge,1),\dots,\aFlag(\anEdge,k)\rangle$ its~\emph{flag pattern}.
	\hbox{\cite{Bau09}} observe for Arc-Flags on road networks that many edges in the graph share the same flag pattern.
	They exploit this with the following compression technique:
	All flag patterns that occur in the graph are stored in a global array $A$.
	For each edge $\anEdge$, the algorithm does not store the flag pattern of $\anEdge$ directly, but rather the index $i$ for which $A[i]$ contains the pattern.
	This saves space at the cost of additional pointer accesses whenever an edge is explored.
	We apply this compression technique in FLASH-TB and sort the flag pattern array in decreasing order of occurrence.
	This ensures that the most commonly accessed flag patterns are stored close together in memory, which increases the likelihood of cache hits.
	If flag compression is not used, the flags are stored in an array $A$ of size $k\absoluteVal{\transfers}$.
	To improve cache efficiency, we store the flags for each cell consecutively.
	Thus, the flag $\aFlag(\aTransfer,i)$ for a transfer $\aTransfer$ and cell $i$ is stored at position $A[i\absoluteVal{\transfers}+\aTransfer]$.
	This ensures that flags for irrelevant target cells do not occupy space in the cache.
	
	\subparagraph*{Reached Index Timestamping.}
	Usually, TB resets the reached index of every trip $\aTrip$ to $\absoluteVal{\aTrip}$ after each query.
	However, the search space of FLASH-TB is so small that the vast majority of reached indices are not modified during a typical query.
	To prevent unnecessary overhead, we store a 16-bit \textit{timestamp} $\eta(\aTrip)$ for each trip $\aTrip$, which indicates last query in which $\reachedIndex(\aTrip)$ was modified.
	When $\reachedIndex(\aTrip)$ is accessed during the $i$-th query and $\eta\left(\aTrip\right) \neq i$, $\reachedIndex(\aTrip)$ is reset to $\absoluteVal{\aTrip}$ and $\eta(\aTrip)$ is set to $i$.
	Every $2^{16}$ queries, the timestamp overflows, so all reached indices are reset.
    
	\section{Comparison with Other Algorithms}
	\label{sec:comparison}
	In this section, we outline the similarities and differences of FLASH-TB to other algorithms.
	In particular, we explain why FLASH-TB reduces the search space compared to TP, TB-CST, and previous applications of Arc-Flags in public transit networks.
	
	\subsection{Transfer Patterns and TB-CST}
	Both TP and TB-CST (without split trees) store a generalized shortest-path tree for every possible source stop.
	This offers a near-perfect reduction in the query search space but at the expense of a memory consumption of $\Omega(\absoluteVal{\stops}^2)$.
	By contrast, the memory consumption of FLASH-TB is in $\Theta\left(\numCells\absoluteVal{\transfers}\right)$, where $\numCells$ is the number of cells.
	Although $\absoluteVal{\transfers}$ is theoretically in $\Omega(\absoluteVal{\stopEvents}\absoluteVal{\trips})$, it is comparable to $\absoluteVal{\stopEvents}$ in practice.
	Thus, FLASH-TB can be seen as a way to interpolate between TB and TP regarding query search space and memory consumption.
	For $\numCells=1$, every non-superfluous transfer is flagged, so the search space is identical to that of TB with a filtered set of transfers.
	For $\numCells=\absoluteVal{\stops}$, \changed{transfers are only flagged if they occur in a Pareto-optimal journey to the target stop that was selected as a representative by canonical Profile-TB}.
	An advantage of our approach is that the transfer flags provide information about which specific trips should be entered.
	By contrast, the search graphs of TB-CST and TP only provide information about lines and stop sequences, respectively.
	This means that FLASH-TB does not have to invest additional effort during the query phase to find the earliest reachable trip of each line.
	
	\subsection{Arc-Flags on Public Transit Networks}
	Conceptually, FLASH-TB is similar to Arc-Flags on a time-expanded graph, albeit with TB as a query algorithm instead of Dijkstra's algorithm.
	\hbox{\cite{Ber09}} and \hbox{\cite{Del09c}} observe low speedups when applying Arc-Flags to public transit networks.
	We analyze the issues causing this and how FLASH-TB overcomes them.
	
	\subparagraph*{Issues with Existing Approaches.}
	\cite{Ber09} apply Arc-Flags to a time-dependent graph model in a problem setting that asks for \emph{all} Pareto-optimal paths, \ie all representatives for every Pareto-optimal cost vector.
	They observe that for nearly every combination of edge $\anEdge$ and cell $i$, there is at least one point in time for which $\anEdge$ occurs on a Pareto-optimal path to a node in cell $i$.
	To solve this problem, the authors divide the service period of the network into two-hour intervals and compute a flag for each combination of edge, cell and time interval.
	However, this merely yields a speedup of three over Dijkstra's algorithm.
	
	\begin{figure}[tp!]
		\caption{
			Two example stops in a time-expanded graph.
			Arrival nodes in lilac, transfer nodes in \changed{dark gray}, and departure nodes in \changed{yellow}.
			Trip edges are thick and colored, while transfer edges are thin and gray.
			Edge weights have been omitted for clarity.
		}
		\label{fig:tegraph}
		\centering
		\input{fig/timeExpandedGraph}
	\end{figure}
	
	Time resolution is not an issue in the time-expanded graph, where each node is associated with a specific point in time.
	Each stop event is modeled with three nodes: an \emph{arrival node}, a \emph{transfer node} and a \emph{departure node}.
	The edges in the graph can be divided into \emph{trip edges}, which connect the arrival and departure nodes of the same trip, and \emph{transfer edges}, which are incident to transfer nodes.
	The transfer nodes of each stop are connected by a ``chain'' of transfer edges leading from earlier to later nodes.
	An example is shown in Figure~\ref{fig:tegraph}.
	A trip segment corresponds to a path between a departure node and an arrival node consisting only of trip edges.
	By contrast, a transfer corresponds to a path of transfer edges between an arrival node and a departure node.
	
	In principle, applying Arc-Flags to the time-expanded graph should yield much more fine-grained information than in the time-dependent model.
	However, \hbox{\cite{Del09c}} still only report a speedup of four over Dijkstra's algorithm, even when optimizing only the arrival time.
	One reason for this is noted by the authors: in a time-expanded graph, every path is optimal because its arrival time depends only on the time of the target event.
	To solve this problem, the authors employ strategies for breaking ties between paths.
    Moreover, they propose a pruning technique called \emph{node-blocking}, which applies the principle of line pruning to Dijkstra's algorithm in time-expanded graphs.
	They observe that node-blocking conflicts with their tiebreaking choices for Arc-Flags, leading to incorrectly answered queries.
	Hence, they disable node-blocking when evaluating Arc-Flags, which increases the search space.
	
	A second reason is that their approach computes flags that are sufficient to find a shortest path between any pair of nodes.
	However, the public transit journey planning problem has two crucial differences to the shortest path problem in the time-expanded graph:
	Firstly, the source node may only be a transfer node.
	Secondly, no target node is supplied, only a target stop.
	The actual target node is the earliest reachable transfer node of the target stop, which only becomes known during the search.
	For these reasons, not every pair of source and target node occurs in an optimal journey.
	Thus, the approach by Delling, Pajor, and Wagner flags more edges than necessary.
	
	\subparagraph*{Improvements in FLASH-TB.}
	Arc-Flags on the time-expanded graph flags both trip and transfer edges, whereas FLASH-TB flags only transfers.
	Flagging an entire TB transfer provides more information than flagging all the transfer edges it consists of.
	Consider an example with four stop events $\aTrip_1[i_1],\aTrip_2[i_2],\aTrip_3[i_3],\aTrip_4[i_4]$ at the same stop $\aStop$ such that their arrival and departure times are increasing in this order.
	Let $\aVertex_1,\aVertex_2,\aVertex_3,\aVertex_4$ be the corresponding transfer nodes.
	Then the ``transfer chain'' of $\aStop$ includes the path $\aPath=\left<\aVertex_1,\aVertex_2,\aVertex_3,\aVertex_4\right>$.
	Assume that FLASH-TB flags the transfers $\transfer{\aTrip_1[i_1]}{\aTrip_3[i_3]}$ and $\transfer{\aTrip_2[i_2]}{\aTrip_4[i_4]}$.
    In Arc-Flags, this corresponds to flagging the entire path $\aPath$.
	However, this creates flagged paths from $\aTrip[i_1]$ to $\aTrip_4[i_4]$ and from $\aTrip_2[i_2]$ and $\aTrip_3[i_3]$, which are akin to also flagging the transfers $\transfer{\aTrip_1[i_1]}{\aTrip_4[i_4]}$ and $\transfer{\aTrip_2[i_2]}{\aTrip_3[i_3]}$.
	Thus, Arc-Flags on the time-expanded graph always explores these transfers, whereas FLASH-TB only does so if they are flagged.
	
	We refrain from flagging trip segments in FLASH-TB because it would not provide any benefit beyond the first round of a FLASH-TB query:
	If a trip segment is not flagged for the target cell, then neither are its incoming or outgoing transfers.
	Thus, an unflagged trip segment can only be entered during the first round, and no further trip segments are reachable from there.
	It would be possible to compute a flag for each trip $\aTrip$ and cell $i$ that indicates whether $\aTrip$ is used for the first trip segment of a required journey to cell $i$.
	However, preliminary experiments have shown that this barely improves query times at the cost of additional memory consumption.
	
	The other crucial difference between FLASH-TB and Arc-Flags on time-expanded graphs is how the flags are computed.
	Arc-Flags is typically employed for solving the shortest path problem on a graph, where any node can occur as the source or target node of a query.
	In this scenario, it is sufficient to compute a backward shortest-path tree from each boundary node.
	Because every edge in the tree is required to answer at least one query, it is not wasteful to flag all of them.
	However, as explained above, not all optimal paths between all pairs of nodes in the time-expanded graph occur in an optimal journey.
	Computing and flagging a backward shortest-path tree from every boundary node $\targetVertex$ is wasteful for two reasons:
	If $\targetVertex$ is a departure or arrival node, optimal paths to $\targetVertex$ correspond to partial journeys, which may not occur in any optimal journey.
	If $\targetVertex$ is a transfer node, it is not necessarily the earliest reachable node of its stop for every leaf in the tree.
	FLASH-TB solves both issues by computing the flags with forward Profile-TB searches from every stop in the network.
	Because this approach is independent of the network partition, it does not need to run searches from boundary nodes that are departure or arrival nodes.
	Secondly, because it uses forward searches, it can easily discard journeys with a suboptimal arrival time at the target stop.
	The downside of this approach is that the preprocessing effort becomes quadratic in the number of stops and no longer scales with the partition size.
	
	Finding a suitable preprocessing algorithm based on backward searches from the boundary nodes remains an open problem.
	Such an algorithm would need to discard enough irrelevant paths to keep the search space small while being compatible with the pruning rules of the query algorithm.
	As observed by~\hbox{\cite{Del09c}}, achieving compatibility with line pruning (or their implementation of it, node-blocking) is particularly challenging.
	Line pruning prefers journeys that use the earliest trip of each line that is still reachable from the source node.
	However, this trip is difficult to determine during a backward search from the target node.
	Delling, Pajor, and Wagner were only able to avoid this issue by disabling node-blocking altogether.
	As seen in Section~\ref{sec:transferset:issues}, a naive implementation of FLASH-TB would run into a similar issue due to the conflict between self-pruning and line pruning.
	However, we resolved this by adjusting the preprocessing algorithm, allowing FLASH-TB to benefit from both pruning rules simultaneously.
	
	\section{Experimental Evaluation}
	\label{sec:experiments}
	We evaluate the performance of FLASH-TB on four real-world public transit networks.
	Our experimental setup is outlined in Section~\ref{sec:experiments:setup}.
	Afterwards, we evaluate the partitioning step and the transfer generation algorithms in Section~\ref{sec:experiments:partitioning}.
	This is followed by the performance of FLASH-TB in Section~\ref{sec:experiments:flashtb}.
	Finally, we compare our algorithm to the state of the art in Section~\ref{sec:experiments:comparison}.
	
	\subsection{Setup}
	\label{sec:experiments:setup}
	Our datasets are taken from GTFS feeds of the public transit networks of Germany~(\url{https://gtfs.de/}), Paris~(\url{https://navitia.io/}), Sweden~(\url{https://trafiklab.se/}) and Switzerland~(\url{https://opentransportdata.swiss/}).
	Details are listed in Table~\ref{tab:datasets}.
	For each network, we extracted the timetable of two consecutive (non-holiday) weekdays to allow for overnight journeys.
	A grouping of the trips into a minimal set of lines can be computed in polynomial time~\hbox{\cite{Ste23}}.
	However, a greedy approach~\hbox{\cite{Bau23}} produced the same (optimal) grouping on all four datasets.
	All precomputations were run on a machine with two~64-core AMD Epyc Rome 7742 CPUs clocked at~2.25\,GHz, with a boost frequency of~3.4\,GHz, 1024\,GB of DDR4-3200 RAM, and 256\,MiB of L3 cache.
	All queries were performed on a machine with two~8-core Intel Xeon Skylake SP Gold 6144 CPUs clocked at~3.5\,GHz, with a boost frequency of~4.2\,GHz, 192\,GiB of DDR4-2666 RAM, and~24.75\,MiB of L3 cache.
	The source code for TB and FLASH-TB~(\url{https://github.com/TransitRouting/FLASH-TB}) was written in~C\raisebox{0.15ex}{\small++} and compiled using GCC with optimizations enabled (\codestyling{-march=native -O3}).
	For the comparison to TB-CST, we used the original source code by~\hbox{\cite{Wit15}}, with the additional fixes discussed in Section~\ref{sec:transferset:issues}.
	The preprocessing phases of TB-CST and FLASH-TB, which run one-to-all Profile-TB from each stop, were parallelized with 128 threads.
	The queries were not parallelized.
	
	\begin{table}[tb!]
		\caption{
			An overview of the networks on which we performed our experiments.
			Stops, stop events, lines, trips, and footpaths are taken from the GTFS datasets.
            For the footpaths, we build the transitive closure.
			Also listed are the results of the TB and \transultra transfer precomputations.
			Preprocessing times are formatted as ${\left[\mathrm{hh}{:}\mathrm{mm}{:}\mathrm{ss}\right]}$.
		}
		\label{tab:datasets}
		\centering
		\begin{tabular}{lrrrr}
			\toprule
			& Germany & Paris & Sweden & Switzerland \\
			\midrule
			Stops					& \numprint{441465}		& \numprint{41757} 		& \numprint{48007} 		& \numprint{30510} \\
			Stop events				& \numprint{30995609}	& \numprint{4636238} 	& \numprint{5647754} 	& \numprint{5197469} \\
			Lines					& \numprint{207801}		& \numprint{9558} 		& \numprint{15627} 		& \numprint{17895} \\
			Trips					& \numprint{1559118}	& \numprint{215526} 	& \numprint{248977} 	& \numprint{336765} \\
			Footpaths				& \numprint{1172464}	& \numprint{445912} 	& \numprint{2118} 		& \numprint{20468} \\[5pt]
			TB transfers			& \numprint{60919879}	& \numprint{23284123}	& \numprint{14771466} 	& \numprint{8481315} \\
			TB prepro.				& \printTime{0}{2}{8}	& \printTime{0}{0}{36}	& \printTime{0}{0}{24} 	& \printTime{0}{0}{16} \\[5pt]
			\transultra transfers	& \numprint{49339366}	& \numprint{15388877}	& \numprint{11960527} 	& \numprint{6420392} \\
			\transultra prepro.		& \printTime{0}{11}{34}	& \printTime{0}{0}{39}	& \printTime{0}{0}{35} 	& \printTime{0}{0}{28} \\
			\bottomrule
		\end{tabular}
	\end{table}
	
	\subsection{Partitioning and Transfer Generation}
	\label{sec:experiments:partitioning}
	We use the KaHIP graph partitioning framework~\hbox{\cite{San13}} (\url{https://github.com/KaHIP/KaHIP}) to partition our networks.
	KaHIP follows the multilevel paradigm, \ie the input graph is coarsened, initially partitioned, and locally improved during uncoarsening.
	In our experiments, better results are obtained when the coarsening is done using clustering rather than edge matching, which is the default setting.
	More specifically, we use the memetic algorithm KaFFPaE with the strong social setting and an imbalance parameter of $5\%$.
	As a time limit, we set ten minutes for all networks regardless of the number $\numCells$ of desired cells.
	In our experiments, higher time limits did not significantly improve the total number of flags set or the average query times.
	
	Results for the TB and \transultra transfer generation steps are listed in Table~\ref{tab:datasets}.
	\transultra reduces the number of transfers by around~20--30\%.
	This comes at the cost of a higher running time, especially on the Germany network.
	However, compared to the overall preprocessing time of FLASH-TB, the difference is negligible.
	
	\subsection{FLASH-TB Performance}
	\label{sec:experiments:flashtb}
	\begin{table}[p]
		\caption{
                Performance of FLASH-TB depending on the number ${\numCells}$ of cells.
			Values for ${\numCells=1}$ are for the original TB algorithm.
			Memory consumption and query times and are measured with (Comp.) and without (Uncomp.) flag compression.
			Query times and scanned trips are averaged over~\numprint{10000} random fixed departure time queries.
			\changed{Also reported are profile query times, averaged over \numprint{1000} random queries with the departure time range spanning the first 24 hours in the timetable.}
			The preprocessing time excludes the partitioning, which is limited to ten minutes in all configurations, and the preprocessing time of \transultra (see Table~\ref{tab:datasets}).
		}
		\label{tab:results}
		\centering
		\begin{tabular*}{\textwidth}{@{\,}l@{\extracolsep{\fill}}r@{\extracolsep{\fill}}r@{\extracolsep{\fill}}r@{\extracolsep{\fill}}r@{\extracolsep{\fill}}r@{\extracolsep{\fill}}r@{\extracolsep{\fill}}r@{\extracolsep{\fill}}r@{\extracolsep{\fill}}r@{\extracolsep{\fill}}r@{\,}}
			\toprule
			\multirow{2}{*}{Network} & \multirow{2}{*}{$\numCells$} & \multirow{2}{*}{\thead{Prepro.\\ $\left[ \mathrm{hh}{:}\mathrm{mm}{:}\mathrm{ss}\right]$}} & \multicolumn{2}{c}{Memory $[\mathrm{MB}]$} & \multicolumn{2}{c}{Query time $[\mu s]$} & \multirow{2}{*}{\thead{Scanned\\trips}} & \multicolumn{2}{c}{Profile time $[\mu s]$} \\
			\cmidrule(){4-5} \cmidrule(){6-7} \cmidrule(){9-10}
			& & & Uncomp. & Comp. & Uncomp. &  Comp. & & \changed{Uncomp.} & \changed{Comp.} \\
			\midrule
			\multirow{9}{*}{\rotatebox{90}{Germany}} & \numprint{1} & -- & -- & -- & \numprint{63687} & -- & \numprint{280586} & \numprint{383747} & -- \\
			& \numprint{64} & \printTime{30}{41}{41} & \numprint{1843} & \numprint{280} & \numprint{1555} & \numprint{2026} & \numprint{4912} & \numprint{8164}  & \numprint{10641} \\
			& \numprint{128} & \printTime{29}{13}{44} & \numprint{2150} & \numprint{307} & \numprint{902} & \numprint{1314} & \numprint{2769} & \numprint{4858}  & \numprint{6996} \\
			& \numprint{256} & \printTime{28}{40}{21} & \numprint{2662} & \numprint{359} & \numprint{547} & \numprint{888} & \numprint{1598} & \numprint{3076}  & \numprint{4961} \\
			& \numprint{512} & \printTime{28}{28}{33} & \numprint{2686} & \numprint{473} & \numprint{355} & \numprint{635} & \numprint{976} & \numprint{2052}  & \numprint{3702} \\
			& \numprint{1024} & \printTime{28}{19}{22} & \numprint{5734} & \numprint{720} & \numprint{234} & \numprint{468} & \numprint{608} & \numprint{1383}  & \numprint{2800} \\
			& \numprint{2048} & \printTime{28}{29}{21} & \numprint{9830} & \numprint{1384} & \numprint{168} & \numprint{372} & \numprint{409} & \numprint{1052}  & \numprint{2329} \\
			& \numprint{4096} & \printTime{28}{40}{08} & \numprint{19251} & \numprint{2716} & \numprint{130} & \numprint{321} & \numprint{291} & \numprint{840}  & \numprint{2839} \\
			& \numprint{8192} & \printTime{32}{36}{15} & \numprint{34816} & \numprint{5788} & \numprint{127} & \numprint{419} & \numprint{218} & \numprint{861}  & \numprint{2590} \\[5pt]
			\multirow{10}{*}{\rotatebox{90}{Paris}} & \numprint{1}	& -- & --	& -- & \numprint{2804} & --	& \numprint{17619} & \numprint{44197} & --\\
			& \numprint{64} & \printTime{0}{18}{41} & \numprint{408} & \numprint{88} & \numprint{314} & \numprint{535} & \numprint{796} & \numprint{6854}  & \numprint{14529}  \\
			& \numprint{128} & \printTime{0}{18}{42} & \numprint{465} & \numprint{121} & \numprint{204} & \numprint{403} & \numprint{493} & \numprint{4733}  & \numprint{11878} \\
			& \numprint{256} & \printTime{0}{18}{43} & \numprint{579} & \numprint{174} & \numprint{139} & \numprint{327} & \numprint{316} & \numprint{3302}  & \numprint{9760} \\
			& \numprint{512} & \printTime{0}{18}{44} & \numprint{808} & \numprint{266} & \numprint{93} & \numprint{241} & \numprint{203} & \numprint{2299}  & \numprint{8035} \\
			& \numprint{1024} & \printTime{0}{18}{46} & \numprint{1331} & \numprint{447} & \numprint{67} & \numprint{199} & \numprint{140} & \numprint{1784}  & \numprint{6424} \\
			& \numprint{2048} & \printTime{0}{18}{49} & \numprint{2253} & \numprint{806} & \numprint{47} & \numprint{155} & \numprint{96} & \numprint{1330}  & \numprint{5069} \\
			& \numprint{4096} & \printTime{0}{18}{52} & \numprint{4096} & \numprint{1594} & \numprint{36} & \numprint{132} & \numprint{72} & \numprint{1111}  & \numprint{4219} \\
			& \numprint{8192} & \printTime{0}{18}{54} & \numprint{7680} & \numprint{3232} & \numprint{29} & \numprint{116} & \numprint{57} & \numprint{930}  & \numprint{4225} \\
			& \numprint{16384} & \printTime{0}{19}{1} & \numprint{15360} & \numprint{6612} & \numprint{25} & \numprint{110} & \numprint{48} & \numprint{920}  & \numprint{5053} \\[5pt]
			\multirow{10}{*}{\rotatebox{90}{Sweden}} & \numprint{1} & -- & -- & -- & \numprint{4926} & -- & \numprint{26850} & \numprint{39281} & --\\
			& \numprint{64} & \printTime{0}{12}{35} & \numprint{403} & \numprint{63} & \numprint{106} & \numprint{133} & \numprint{444} & \numprint{829}  & \numprint{1166} \\
			& \numprint{128} & \printTime{0}{12}{18} & \numprint{460} & \numprint{70} & \numprint{62} & \numprint{85} & \numprint{244} & \numprint{503}  & \numprint{728} \\
			& \numprint{256} & \printTime{0}{12}{07} & \numprint{573} & \numprint{81} & \numprint{41} & \numprint{61} & \numprint{142} & \numprint{376}  & \numprint{591} \\
			& \numprint{512} & \printTime{0}{10}{50} & \numprint{799} & \numprint{105} & \numprint{28} & \numprint{48} & \numprint{87} & \numprint{289}  & \numprint{501} \\
			& \numprint{1024} & \printTime{0}{10}{19} & \numprint{1331} & \numprint{160} & \numprint{21} & \numprint{42} & \numprint{58} & \numprint{247}  & \numprint{463}\\
			& \numprint{2048} & \printTime{0}{10}{8} & \numprint{2253} & \numprint{287} & \numprint{17} & \numprint{40} & \numprint{42} & \numprint{204}  & \numprint{453} \\
			& \numprint{4096} & \printTime{0}{10}{9} & \numprint{3994} & \numprint{579} & \numprint{14} & \numprint{41} & \numprint{34} & \numprint{184}   & \numprint{459} \\
			& \numprint{8192} & \printTime{0}{10}{11} & \numprint{7680} & \numprint{1286} & \numprint{13} & \numprint{44} & \numprint{28} & \numprint{174}  & \numprint{468} \\
			& \numprint{16384} & \printTime{0}{10}{14} & \numprint{15360} & \numprint{2719} & \numprint{11} & \numprint{45} & \numprint{25} & \numprint{167}  & \numprint{546} \\[5pt]
			\multirow{10}{*}{\rotatebox{90}{Switzerland}} & \numprint{1} & -- & -- & -- & \numprint{5005} & -- & \numprint{35951} & \numprint{36368}  & -- \\
			& \numprint{64} & \printTime{0}{11}{29} & \numprint{204} & \numprint{31} & \numprint{92} & \numprint{106} & \numprint{475} & \numprint{907} & \numprint{1118} \\
			& \numprint{128} & \printTime{0}{9}{45} & \numprint{233} & \numprint{34} & \numprint{58} & \numprint{71} & \numprint{278} & \numprint{617}  & \numprint{801} \\
			& \numprint{256} & \printTime{0}{8}{12} & \numprint{290} & \numprint{42} & \numprint{39} & \numprint{50} & \numprint{165} & \numprint{439}  & \numprint{610} \\
			& \numprint{512} & \printTime{0}{7}{18} & \numprint{404} & \numprint{58} & \numprint{28} & \numprint{40} & \numprint{111} & \numprint{346}  & \numprint{525} \\
			& \numprint{1024} & \printTime{0}{6}{52} & \numprint{633} & \numprint{96} & \numprint{23} & \numprint{36} & \numprint{83} & \numprint{298} & \numprint{497} \\
			& \numprint{2048} & \printTime{0}{6}{44} & \numprint{1126} & \numprint{183} & \numprint{19} & \numprint{34} & \numprint{66} & \numprint{264}  & \numprint{470}\\
			& \numprint{4096} & \printTime{0}{6}{42} & \numprint{2048} & \numprint{388} & \numprint{17} & \numprint{33} & \numprint{55} & \numprint{243} & \numprint{460} \\
			& \numprint{8192} & \printTime{0}{6}{43} & \numprint{3891} & \numprint{832} & \numprint{15} & \numprint{32} & \numprint{49} & \numprint{229} & \numprint{457} \\
			& \numprint{16384} & \printTime{0}{6}{47} & \numprint{7578} & \numprint{1872} & \numprint{15} & \numprint{33} & \numprint{46} & \numprint{222}  & \numprint{466}\\
			\bottomrule
		\end{tabular*}
	\end{table}

	\begin{figure}[tp!]
		\caption{
			Average speedup over TB and memory consumption of FLASH-TB (with and without flag compression) on the Paris and Switzerland networks, depending on the number of cells ${\numCells}$.
		The partition of Switzerland with ${\numCells=\absoluteVal{\stops}}$, in which each stop has its own cell, was computed directly without the use of KaHIP.}
		\label{fig:resultplot}
		\centering
		\input{fig/speedupPlot}
		\vspace*{-20pt}
	\end{figure}

	Performance measurements for FLASH-TB, including the impact of flag compression and the number $\numCells$ of cells, are shown in Table~\ref{tab:results}.
	For each configuration, we ran \numprint{10000} queries with the source and target stops chosen uniformly at random and the departure time chosen uniformly at random within the first day covered in the network.
	Values for $\numCells$ that do not substantially improve the memory consumption or the query time compared to other values are omitted from Table~\ref{tab:results}.
	For the Germany network, we suspect that there would be further performance gains with $\numCells=\numprint{16384}$ cells, but the machine ran out of memory during the precomputation with~128 threads.
	Additionally, for the Paris and Switzerland networks, Figure~\ref{fig:resultplot} plots the speedup over TB \changed{(for the fixed-departure time query)} and the memory consumption, with and without flag compression, depending on $\numCells$.

        \changed{
        The reported query times exclude the time required to unpack the journeys.
        For TB, journey unpacking took $5\,\mu s$ on Paris, $7\,\mu s$ on Sweden, $11\,\mu s$ on Switzerland, and $17\, \mu s$ for Germany.
        For FLASH-TB with the highest evaluated choice of $k$, this is reduced to $2\,\mu s$ on Sweden and Paris, $3\,\mu s$ on Switzerland, and $9\,\mu s$ on Germany.
        The speedup is due to improved cache locality, as the queue of scanned trip segments is much smaller.        
         }

	As expected, the preprocessing time of FLASH-TB is independent of $\numCells$.
	With the time for partitioning the network and the \transultra transfer precomputation included, the preprocessing takes less than~30 minutes for the three smaller networks.
	The much higher times for Germany indicate that the precomputation effort is quadratic in the network size.
	
	Flagging the transfers is extremely effective at pruning the search space:
	With the highest number of cells, FLASH-TB reduces the number of scanned trips by a factor of~367 on the Paris network,~782 on Switzerland and over~1\,000 on the other two networks.
	Without flag compression, this yields speedups of~501 on Germany, 112 on Paris, 448 on Sweden and 334 on Switzerland.
	Because the memory consumption scales linearly with the number of cells, it quickly becomes very high, whereas the speedup eventually begins to stall.
	However, excellent query times can be achieved with a reasonable amount of data:
	For a speedup of 100 over TB, FLASH-TB without memory compression requires a few gigabytes on Germany and Paris, and a few hundred megabytes on Sweden and Switzerland.
	For profile queries, the speedups are slightly lower at~457 for Germany, 48 for Paris, 235 for Sweden, and 164 for Switzerland.
	This is because self-pruning already thins out the search space of each individual run, which means that there is less potential for additional pruning.

	Flag compression reduces the memory consumption by a factor of two to five on Paris and four to seven on the other networks.
	The resulting loss in query speed ranges from almost nothing to a factor of four; it is generally greater for higher values of $\numCells$.
	This is because the performance overhead caused by the additional memory accesses amortizes if cached flag patterns are accessed frequently.
	Once the search space is reduced to a few dozen trips per query, this is no longer the case.
	With flag compression, the highest achieved speedups are 198 on Germany, 25 on Paris, 123 on Sweden, and 156 on Switzerland.
	The space required for a speedup of~100 over TB is reduced to around~500\,MB for Germany and around or less than~100\,MB for Sweden and Switzerland.
	
	We observe that Paris, a dense metropolitan network, behaves differently from the other networks, which represent entire countries.
	The overall speedup of FLASH-TB is lower, and in relation to the size of the network, the memory consumption is much higher.
	Whereas a speedup of~100 can be achieved with around~10~bytes per transfer on the other networks, it requires~500 for Paris.
	This is both because more cells are required and because flag compression is less effective.
	Similar discrepancies in the performance between metropolitan and country networks were observed for Transfer Patterns~\hbox{\cite{Bas10}} and TB-CST~\hbox{\cite{Wit16}}.
	
	\begin{figure}[tp!]
		\caption{
			Flag distributions for ${\numCells=32}$ cells.
			The bar for value ${i}$ represents the percentage of transfers for which exactly ${i}$ out of ${\numCells}$ flags are set.
		}
		\label{fig:histogram}
		\centering
		\input{fig/flagDistribution}
	\end{figure}
	
	To explain these differences, we analyze how the number of set flags is distributed across the set of transfers.
	Figure~\ref{fig:histogram} shows a histogram of the flag distribution for a partition with $\numCells=32$ cells.
	Depending on the network, between~30\% and~50\% of transfers have no flags set at all and are therefore removed at the end of the preprocessing phase.
	On all networks except for Paris, over~70\% of the remaining transfers have only one flag set and over~80\% have at most two.
	There is also a small group of transfers with all or almost all flags set, but there are very few in between.
	This explains why FLASH-TB is so effective: most transfers are only relevant for local, intra-cell travel and can therefore be ignored outside of the target cell.
	For Paris, this effect is weaker because the network consists almost entirely of local transport.
	Accordingly, more flags are set overall and fewer transfers share the same pattern, so flag compression is less effective.
	
	\subsection{Comparison to TB-CST}
	\label{sec:experiments:comparison}
	We compare FLASH-TB to TB-CST, which is the state of the art among speedup techniques with a high preprocessing effort.
	A direct comparison with TP is difficult because we do not have access to its original implementation.
	However, preliminary experiments with our own implementation of TP confirmed the findings of~\hbox{\cite{Wit16}} that TB-CST without split trees has a very similar memory consumption and query performance while the preprocessing time is significantly lower.
	We also do not compare our algorithm to Scalable Transfer Patterns, whose query times are barely competitive with TB, or to PTL, which does not support journey unpacking.
	
	The results of our evaluation are shown in Table~\ref{tab:tbcst-results}.
	The precomputation of TB-CST is four to ten times faster than that of FLASH-TB, even though both algorithms solve the all-to-all problem.
	This is because our algorithm cannot fully benefit from self-pruning due to the ULTRA conditions.
	However, the additional preprocessing effort pays off in the query phase.
	Whereas TB-CST must reconstruct the earliest reachable trip of each line during the query, FLASH-TB can access it directly.
	Furthermore, TB-CST has an additional time overhead for constructing the query graph.
	Consequently, it is significantly outperformed for fixed departure time queries:
	When comparing TB-CST with split trees to a configuration of FLASH-TB that requires approximately the same amount of space, FLASH-TB is two to nine times faster.
	The speedup is especially high for Germany and Paris.
	On the other hand, to achieve the same query speed, FLASH-TB requires between one and two orders of magnitude less space.
	Additionally, the fastest configuration of FLASH-TB outperforms TB-CST without split trees while still requiring much less space.
	
	For profile queries, TB-CST loses much of its disadvantage.
	Firstly, the construction of the query graph takes up a smaller portion of the overall running time.
	Secondly, instead of performing a run for each possible departure time, TB-CST constructs the entire profile during a single search within the query graph.
	Thus, the step of finding the earliest reachable trip of a line is performed for all profile entries simultaneously, making it more efficient.
	Consequently, TB-CST answers profile queries slightly faster than FLASH-TB.
	However, we note that our implementation of FLASH-TB is not specifically optimized for profile queries and its performance can likely be improved.
	
	\begin{table}[t]
		\centering
		\caption{
			Performance of TB-CST for \numprint{10000} random fixed departure time queries and \numprint{1000} random profile queries with a departure time range of 24 hours. Query times are split into the time for building the query graph (\emph{Build}) and exploring it (\emph{Search}). On the Germany network, TB-CST without split trees ran out of memory.
		}
		\label{tab:tbcst-results}
		\begin{tabular*}{\textwidth}{@{\,}l@{\extracolsep{\fill}}c@{\extracolsep{\fill}}r@{\extracolsep{\fill}}r@{\extracolsep{\fill}}r@{\extracolsep{\fill}}r@{\extracolsep{\fill}}r@{\extracolsep{\fill}}r@{\extracolsep{\fill}}r@{\extracolsep{\fill}}r@{\,}}
			\toprule
			\multirow{2}{*}{Network} & \multirow{2}{*}{\thead{Split\\trees?}} & \multirow{2}{*}{\thead{Prepro.\\ $\left[ \mathrm{hh}{:}\mathrm{mm}{:}\mathrm{ss}\right]$}} & \multicolumn{3}{c}{Fixed dep. query $\left[ \mu s\right]$} & \multicolumn{3}{c}{Profile query $\left[ \mu s\right]$} & \multirow{2}{*}{Memory $\left[ \mathrm{MB} \right]$} \\
			\cmidrule(){4-6} \cmidrule(){7-9}
			& & & Build & Search & Total & Build & Search & Total & \\
			\midrule
			\multirow{2}{*}{Germany} & $\circ$ & -- & -- & -- & -- & -- & -- & -- & OOM \\
			& $\bullet$ & \printTime{9}{27}{42} & \numprint{576} & \numprint{441} & \numprint{1017} & \numprint{543} & \numprint{609} & \numprint{1152} & \numprint{37786} \\[5pt]
			\multirow{2}{*}{Paris} & $\circ$ & \printTime{0}{4}{0} & \numprint{8} & \numprint{28} & \numprint{36} & \numprint{8} & \numprint{151} & \numprint{159} & \numprint{103928} \\
			& $\bullet$ & \printTime{0}{6}{17} & \numprint{68} & \numprint{359} & \numprint{427} & \numprint{65} & \numprint{971} & \numprint{1036} & \numprint{2822} \\[5pt]
			\multirow{2}{*}{Sweden} & $\circ$ & \printTime{0}{1}{52} & \numprint{9} & \numprint{11} & \numprint{20} & \numprint{11} & \numprint{40} & \numprint{51} & \numprint{111370} \\
			& $\bullet$ & \printTime{0}{3}{02} & \numprint{15} & \numprint{32} & \numprint{46} & \numprint{15} & \numprint{83} & \numprint{98} & \numprint{1350} \\[5pt]
			\multirow{2}{*}{Switzerland} & $\circ$ & \printTime{0}{1}{25} & \numprint{24} & \numprint{19} & \numprint{44} & \numprint{23} & \numprint{73} & \numprint{96} & \numprint{86060}\\
			& $\bullet$ & \printTime{0}{2}{30} & \numprint{37} & \numprint{40} & \numprint{77} & \numprint{38} & \numprint{118} & \numprint{156} & \numprint{1559} \\
			\bottomrule
		\end{tabular*}
	\end{table}
	
	\section{Conclusion}
	\label{sec:conclusion}
    \changed{FLASH-TB constitutes a step towards making extremely fast public transit journey planning more practical.
    It demonstrates for the first time that a moderate amount of precomputed data is sufficient to achieve sub-millisecond query times even on country-sized networks.
    The main challenge that must be overcome for this is that the schedule-based nature of public transit leads to a highly redundant path structure, in which many journeys are equivalent in all criteria and optimal paths between stop events do not necessarily appear in any optimal journey.
    If this redundant information is not discarded in a consistent manner, then either the speedup will degrade significantly or some queries will be answered incorrectly.}

    \changed{Despite these advances, FLASH-TB in its current form is still of limited use to real-world journey planning applications.
    The main drawback is the quadratic preprocessing time, which is prohibitive for the largest networks.
    Some existing algorithms, including Scalable TP and Arc-Flags on time-expanded graphs, reduce the preprocessing time by splitting the work into intra-cell and inter-cell computations, but this reintroduces redundancies and therefore causes a significant blowup in the search space.
    Thus, the most important open question is whether this blowup can be mitigated.
    Tackling this question may require further structural insights.}
    
    Another important issue in practice is handling delays.
	The search graphs computed by TP are known to be very robust in this regard~\cite{Bas13b}.
	By comparison, FLASH-TB is likely to be more sensitive to delays because its pruning information is highly time-dependent.
	It remains to be shown whether a delay-resistant variant of FLASH-TB can be designed that still retains most of its speedup over TB.
	Finally, real-world timetables cover more than two days, \eg an entire year.
    Currently, the preprocessing effort of FLASH-TB is proportional to the number of days.
	However, \cite{witt2021extendingtimehorizonefficient} presents a TB variant that handles longer timetables efficiently.
    Incorporating similar ideas into the FLASH-TB preprocessing step may improve its scalability for longer timetables.

    \appendix

		\section{Trans-ULTRA Optimizations}
		\label{app:optimizations}
		A simplified version of \transultra was presented in Section~\ref{sec:transferset:ultra}.
		In this section, we present additional performance optimizations that are carried over from ULTRA.
		These are based on the observation that in order to generate a \changed{CP} transfer set, it is not necessary to find and unpack all \changed{canonical} journeys, merely those that are candidates.
		
		In the profile search performed from each source stop $\sourceStop$, this is exploited by skipping all runs that cannot produce any candidates.
		The set of initial stop events that can occur in a journey with departure time $\departureTime$ is given by
		\[\stopEvents(\sourceStop,\departureTime):=\{ \aTrip[i] \mid \departureTime+\transfertime{\sourceStop}{\aStop(\aTrip[i])}=\departureTime(\aTrip[i])\}. \]
		Normally, the profile search performs a run for every departure time $\departureTime$ such that $\stopEvents(\sourceStop,\departureTime)$ is not empty.
		We call $\departureTime$ a \emph{candidate departure time} if there is a stop event $\aTrip[i]\in\stopEvents(\sourceStop,\departureTime)$ with $\aStop(\aTrip[i])=\sourceStop$.
		Let $\departureTime^1,\dots,\departureTime^k$ be the sequence of candidate departure times, sorted in ascending order.
		To simplify the notation, let $\departureTime^{k+1}:=\infty$.
		Trans-ULTRA performs a run for each candidate departure time $\departureTime^i$ in descending order; the non-candidate departure times are skipped.
		As a consequence, the run for $\departureTime^i$ may find new journeys whose departure time is not exactly $\departureTime^i$ but lies in the interval $[\departureTime^i,\departureTime^{i+1})$.
		The proof of Lemma~\ref{th:departure_canonical_mr} still carries over to a weaker claim, which is sufficient to prove Theorem~\ref{th:ultra:pdop}:
		\begin{lemma}\label{th:departure_canonical_candidate}
			Let $\aJourney$ be a candidate that is \changed{canonical} for the query $\query=(\sourceStop(\aJourney),\targetStop(\aJourney),\departureTime(\aJourney))$.
			After the run for $\departureTime(\aJourney)$ in the \changed{canonical} rRAPTOR search from $\sourceStop(\aJourney)$, the representative for $\cost(\aJourney)$ at $\targetStop(\aJourney)$ is $\aJourney$.
		\end{lemma}
		
		Another optimization is that the algorithm only maintains parent pointers for representatives whose initial transfer is empty; all other pointer are set to a dummy value.
		This allows the algorithm to efficiently identify whether a found journey is a candidate and needs to be unpacked.
		It can also be leveraged by adding an extra ULTRA condition (U3):
		Let $\aJourney$ be the newly found journey and let $\aJourney'$ be the previous representative.
		Then $\aJourney$ is discarded if $\aJourney' \preceq \aJourney$ and $\aJourney$ has a non-empty initial transfer.
		This does not affect the correctness of Lemma~\ref{th:departure_canonical_candidate} because~(U3) is not fulfilled for any prefix of a \changed{canonical} candidate.

\bibliography{bibliography.bib}
\end{document}

%% file: include/notation.tex
\newcommand{\absoluteVal}[1]{\left\vert #1 \right\vert}

\newcommand{\atime}{\ensuremath{\uptau}\xspace}
\newcommand{\departureTime}{\atime_\mathrm{dep}}
\newcommand{\arrivalTime}{\atime_\mathrm{arr}}

\newcommand{\graph}{G}
\newcommand{\vertices}{V}
\newcommand{\edges}{E}
\newcommand{\aVertex}{v}
\newcommand{\vertexA}{u}
\newcommand{\vertexB}{v}
\newcommand{\sourceVertex}{\aVertex_s}
\newcommand{\targetVertex}{\aVertex_t}
\newcommand{\anEdge}{e}
\newcommand{\edgeWeight}{c}
\newcommand{\aPath}{P}

\newcommand{\stops}{\mathcal{S}}
\newcommand{\aStop}{p}
\newcommand{\stopA}{p}
\newcommand{\stopB}{q}
\newcommand{\stopC}{r}
\newcommand{\inStop}{\aStop^\text{in}}
\newcommand{\outStop}{\aStop^\text{out}}
\newcommand{\sourceStop}{\aStop_s}
\newcommand{\targetStop}{\aStop_t}
\newcommand{\stopEvents}{\mathcal{E}}
\newcommand{\stopEvent}{\varepsilon}
\newcommand{\targetEvent}{\stopEvent_t}
\newcommand{\lines}{\mathcal{L}}
\newcommand{\aLine}{L}
\newcommand{\trips}{\mathcal{T}}
\newcommand{\aTrip}{T}
\newcommand{\tripA}{T_a}
\newcommand{\tripB}{T_b}
\newcommand{\tripC}{T_c}
\newcommand{\tripD}{T_d}

\newcommand{\pred}[1]{\mathrm{pred}\left( #1 \right)}
\newcommand{\tripSegment}[3]{#1[#2,#3]}
\newcommand{\footpaths}{\mathcal{F}}
\newcommand{\footpath}{f}
\newcommand{\transfertime}[2]{\Delta\atime_{\mathrm{fp}}( #1, #2 )}
\newcommand{\query}{Q}

\newcommand{\journeys}{\mathcal{J}}
\newcommand{\aJourney}{J}
\newcommand{\subjourney}[2]{\aJourney[#1,#2]}
\newcommand{\subjourneyAlt}[2]{\aJourney'[#1,#2]}
\newcommand{\transfers}{\mathfrak{T}}
\newcommand{\aTransfer}{\mathfrak{t}}
\newcommand{\transfer}[2]{(#1,#2)}

\newcommand{\cost}{c}
\newcommand{\costspace}{\mathcal{C}}
\newcommand{\costsystem}{C}
\newcommand{\paretoset}{\mathcal{P}}
\newcommand{\representativeSet}{\journeys_R}

\newcommand{\numCells}{k}
\newcommand{\partition}{r}
\newcommand{\aFlag}{b}

\newcommand{\layoutGraph}{\graph_L}
\newcommand{\layoutEdges}{\edges_L}
\newcommand{\layoutEdgeWeight}{\edgeWeight_L}
\newcommand{\links}{C}

\newcommand{\aQueue}{\ensuremath{Q}\xspace}
\newcommand{\reachedIndex}{R}
\newcommand{\runReachedIndex}{R_{\mathrm{r}}}
\newcommand{\minTime}{\ensuremath{\atime_{\textsf{min}}}\xspace}
\SetKwFunction{Enqueue}{Enqueue}
\SetKwFunction{Scan}{Scan}
\SetKwFunction{AddArrival}{AddArrival}

\newcommand{\tiebreakingSequence}{X}
\newcommand{\tiebreakers}{\tiebreakingSequence_b}
\newcommand{\localTiebreakingSequence}{\tiebreakingSequence_\ell}

\newcommand{\lineIndex}{\ensuremath{\text{id}_\lines}\xspace}
\newcommand{\stopIndex}{\ensuremath{\text{id}_\stops}\xspace}
\newcommand{\eventIndex}{\ensuremath{\text{id}_\stopEvents}\xspace}
\newcommand{\directTransfers}{\ensuremath{\transfers_0}\xspace}
\newcommand{\footpathTransfers}{\ensuremath{\transfers_\mathrm{fp}}\xspace}
\DeclareMathOperator{\parent}{parent}

%% file: include/colors.tex
\usepackage{xcolor}
\definecolor{KITgreen}          {rgb}{0,    0.588,0.509}
\definecolor{KITgreen70}        {rgb}{0.3,  0.711,0.656}
\definecolor{KITgreen50}        {rgb}{0.5,  0.794,0.754}
\definecolor{KITgreen30}        {rgb}{0.7,  0.876,0.852}
\definecolor{KITgreen15}        {rgb}{0.85, 0.938,0.926}
\definecolor{KITblue}           {rgb}{0.274,0.392,0.666}
\definecolor{KITblue70}         {rgb}{0.492,0.574,0.766}
\definecolor{KITblue50}         {rgb}{0.637,0.696,0.833}
\definecolor{KITblue30}         {rgb}{0.782,0.817,0.9}
\definecolor{KITblue15}         {rgb}{0.891,0.908,0.95}
\definecolor{KITpalegreen}      {rgb}{0.509,0.745,0.235}
\definecolor{KITpalegreen70}    {rgb}{0.656,0.821,0.464}
\definecolor{KITpalegreen50}    {rgb}{0.754,0.872,0.617}
\definecolor{KITpalegreen30}    {rgb}{0.852,0.923,0.77}
\definecolor{KITpalegreen15}    {rgb}{0.926,0.961,0.885}
\definecolor{KITyellow}         {rgb}{0.98, 0.901,0.078}
\definecolor{KITyellow70}       {rgb}{0.986,0.931,0.354}
\definecolor{KITyellow50}       {rgb}{0.99, 0.95, 0.539}
\definecolor{KITyellow30}       {rgb}{0.994,0.97, 0.723}
\definecolor{KITyellow15}       {rgb}{0.997,0.985,0.861}
\definecolor{KITorange}         {rgb}{0.862,0.627,0.117}
\definecolor{KITorange70}       {rgb}{0.903,0.739,0.382}
\definecolor{KITorange50}       {rgb}{0.931,0.813,0.558}
\definecolor{KITorange30}       {rgb}{0.958,0.888,0.735}
\definecolor{KITorange15}       {rgb}{0.979,0.944,0.867}
\definecolor{KITbrown}          {rgb}{0.627,0.509,0.196}
\definecolor{KITbrown70}        {rgb}{0.739,0.656,0.437}
\definecolor{KITbrown50}        {rgb}{0.813,0.754,0.598}
\definecolor{KITbrown30}        {rgb}{0.888,0.852,0.758}
\definecolor{KITbrown15}        {rgb}{0.944,0.926,0.879}
\definecolor{KITred}            {rgb}{0.627,0.117,0.156}
\definecolor{KITred70}          {rgb}{0.739,0.382,0.409}
\definecolor{KITred50}          {rgb}{0.813,0.558,0.578}
\definecolor{KITred30}          {rgb}{0.888,0.735,0.747}
\definecolor{KITred15}          {rgb}{0.944,0.867,0.873}
\definecolor{KITlilac}          {rgb}{0.627,0,    0.47}
\definecolor{KITlilac70}        {rgb}{0.739,0.3,  0.629}
\definecolor{KITlilac50}        {rgb}{0.813,0.5,  0.735}
\definecolor{KITlilac30}        {rgb}{0.888,0.7,  0.841}
\definecolor{KITlilac15}        {rgb}{0.944,0.85, 0.92}
\definecolor{KITcyanblue}       {rgb}{0.313,0.666,0.901}
\definecolor{KITcyanblue70}     {rgb}{0.519,0.766,0.931}
\definecolor{KITcyanblue50}     {rgb}{0.656,0.833,0.95}
\definecolor{KITcyanblue30}     {rgb}{0.794,0.9,  0.97}
\definecolor{KITcyanblue15}     {rgb}{0.897,0.95, 0.985}
\definecolor{KITseablue}        {rgb}{0.196,0.313,0.549}
\definecolor{KITseablue70}      {rgb}{0.437,0.519,0.684}
\definecolor{KITseablue50}      {rgb}{0.598,0.656,0.774}
\definecolor{KITseablue30}      {rgb}{0.758,0.794,0.864}
\definecolor{KITseablue15}      {rgb}{0.879,0.897,0.932}
\definecolor{KITblack}          {rgb}{0,    0,    0}
\definecolor{KITblack90}        {rgb}{0.1,  0.1,  0.1}
\definecolor{KITblack80}        {rgb}{0.2,  0.2,  0.3}
\definecolor{KITblack75}        {rgb}{0.25, 0.25, 0.25}
\definecolor{KITblack70}        {rgb}{0.3,  0.3,  0.3}
\definecolor{KITblack60}        {rgb}{0.4,  0.4,  0.4}
\definecolor{KITblack50}        {rgb}{0.5,  0.5,  0.5}
\definecolor{KITblack40}        {rgb}{0.6,  0.6,  0.6}
\definecolor{KITblack30}        {rgb}{0.7,  0.7,  0.7}
\definecolor{KITblack25}        {rgb}{0.75, 0.75, 0.75}
\definecolor{KITblack20}        {rgb}{0.8,  0.8,  0.8}
\definecolor{KITblack10}        {rgb}{0.9,  0.9,  0.9}
\definecolor{KITwhite}          {rgb}{1,    1,    1}

%% file: include/figures.tex
\usepackage{tikz}
\usepackage{pgfplots}
\usepackage{colortbl}
\usepgfplotslibrary{groupplots}
\usetikzlibrary{positioning,shapes,arrows,fit,backgrounds}
\pgfplotsset{compat=newest}
\pgfkeys{/pgf/number format/.cd, set thousands separator={\,}}

\colorlet{nodeColor}{black!80}
\colorlet{edgeColor}{black!50}
\colorlet{axisColor}{black!80}
\colorlet{legendColor}{black!80}

\newcommand{\gs}{\hphantom{\tiny$\cdot$}}
\newcommand{\legend}[1]{\raisebox{0.07ex}{#1}}

\tikzstyle{vertex}=[circle,line width=.5pt,minimum size=0.1pt]
\tikzstyle{arrow}=[->, >=stealth]
\tikzstyle{route}=[arrow, line width=2.5pt, rounded corners = 20]
\tikzstyle{edge}=[edgeColor, line width=1pt, rounded corners = 20]
\tikzstyle{directedEdge}=[edge, arrow]
\tikzstyle{stop}=[rectangle, line width=.5pt, inner sep=5pt, draw=nodeColor!100,fill=nodeColor!15,rounded corners=0.1cm]

\pgfplotsset{
    axis line style={axisColor,line width = 0.4pt},
    major tick style={axisColor,line width = 0.4pt},
    minor tick style={axisColor,line width = 0.2pt},
    major tick length=3.5pt,
    minor tick length=2.0pt,
    ytick align=outside,
    xtick align=outside,
    xtick pos=left,
    ytick pos=left,
    ticklabel style = {font=\small},
    label style = {font=\small},
}

%% file: fig/LatestExit.tex
\begin{tikzpicture}
	\node (s0A) at (0.00, 0.00) {};%
	\node (s1A) at (3.00, 0.00) {};%
	\node (s2A) at (6.00, 0.00) {};%
	\node (s3A) at (9.00, 0.00) {};%
	
	\node (s0B) at (7.00, -1.50) {};%
	\node (s1B) at (10.00, -1.50) {};%
    
    \begin{pgfonlayer}{background}
    \node [fit=(s0A),stop] {};%
    \node [fit=(s1A),stop] (v1) {};%
    \node [fit=(s2A),stop] (v2) {};%
    \node [fit=(s3A),stop] {};%
    \node [fit=(s0B),stop] {};%
    \node [fit=(s1B),stop] {};%
	\end{pgfonlayer}
    
	\node [align=left,text=KITblue] at (1.50, 0.30) {\small{$0\rightarrow10$}};%
	\node [align=left,text=KITblue] at (4.50, 0.30) {\small{$10\rightarrow20$}};%
	\node [align=left,text=KITblue] at (7.50, 0.30) {\small{$20\rightarrow30$}};%
	\node [align=left,text=KITgreen] at (8.50, -1.20) {\small{$25\rightarrow35$}};%
	
	\node [align=left,text=KITblue] at (-0.70, 0.00) {\small{$\tripA$}};%
	\node [align=left,text=KITgreen] at (6.30, -1.50) {\small{$\tripB$}};%
	\node [align=left,text=edgeColor] at (4.50, -1.10) {\small{$5$}};%
	\node [align=left,text=edgeColor] at (6.75, -0.70) {\small{$2$}};%
	
	\node (s0A_v) at (s0A) [vertex,draw=KITblue,fill=KITblue!15] {\gs};%
	\node (s1A_v) at (s1A) [vertex,draw=KITblue,fill=KITblue!15] {\gs};%
	\node (s2A_v) at (s2A) [vertex,draw=KITblue,fill=KITblue!15] {\gs};%
	\node (s3A_v) at (s3A) [vertex,draw=KITblue,fill=KITblue!15] {\gs};%
	\node (s0B_v) at (s0B)  [vertex,draw=KITgreen,fill=KITgreen!15] {\gs};%
	\node (s1B_v) at (s1B)  [vertex,draw=KITgreen,fill=KITgreen!15] {\gs};%
    
    \begin{pgfonlayer}{background}
    \draw [KITblue, route] (s0A) -- (s1A) -- (s2A) -- (s3A_v);
    \draw [KITgreen, route] (s0B) -- (s1B_v);
    \draw [directedEdge, shorten <= -1.3pt, shorten >= -1.3pt]  (v1) to[bend right=45] (v2);
    \draw [directedEdge]  (s2A) -- (s0B_v);
    \end{pgfonlayer}
	
	\node at (s0A) [text=KITblue] {\small{$0$}};%
	\node at (s1A) [text=KITblue] {\small{$1$}};%
	\node at (s2A) [text=KITblue] {\small{$2$}};%
	\node at (s3A) [text=KITblue] {\small{$3$}};%
	
	\node at (s0B) [text=KITgreen] {\small{$0$}};%
	\node at (s1B) [text=KITgreen] {\small{$1$}};%
\end{tikzpicture}

%% file: fig/departureTimeBuffering.tex
\begin{tikzpicture}[scale=1.04]
	\begin{scope}[yscale=0.8]
		\node (ta0) at (0.00, 0.00) {};%
		\node (ta1) at (3.00, 0.00) {};%
        \node (tb0) at (0.00, -1.50) {};%
        \node (tb1) at (3.00, -1.50) {};%
		\node (tc0) at  (5.00, 0.00) {};%
		\node (tc1) at  (8.00, 0.00) {};%
        \node (td0) at  (5.00, -1.50) {};%
        \node (td1) at  (8.00, -1.50) {};%
		\node (source_label) at (0.00, 0.50) {};%
		\node (target_label) at (8.00, 0.50) {};%

        \begin{pgfonlayer}{background}
		\node [fit=(ta0)(tb0)(source_label),stop] {};%
		\node [fit=(ta1)(tb1),stop] {};%
		\node [fit=(tc0),stop] {};%
        \node [fit=(td0),stop] {};%
		\node [fit=(tc1)(td1)(target_label),stop] {};%
		\end{pgfonlayer}
		
		\node [align=left,text=KITblue] at ( 0.60, -0.40) {\small{$\tripA$}};%
		\node [align=left,text=KITblue] at ( 0.60, -1.90) {\small{$\tripB$}};%
		\node [align=left,text=KITgreen] at ( 5.60, -0.40) {\small{$\tripC$}};%
		\node [align=left,text=KITred] at ( 5.60, -1.90) {\small{$\tripD$}};%
        
        \node [align=left,text=KITblue] at (1.50, 0.30) {\small{$5\rightarrow10$}};%
		\node [align=left,text=KITblue] at (1.50, -1.20) {\small{$0\rightarrow5$}};%
        \node [align=left,text=KITgreen] at (6.50, 0.30) {\small{$15\rightarrow20$}};%
        \node [align=left,text=KITred] at (6.50, -1.20) {\small{$10\rightarrow20$}};%
        
		\node (ta0_v) at (ta0) [vertex,draw=KITblue,fill=KITblue!15] {\gs};%
		\node (tb0_v) at (tb0) [vertex,draw=KITblue,fill=KITblue!15] {\gs};%
		\node (ta1_v) at (ta1) [vertex,draw=KITblue,fill=KITblue!15] {\gs};%
		\node (tb1_v) at (tb1) [vertex,draw=KITblue,fill=KITblue!15] {\gs};%
		\node (tc0_v) at (tc0)  [vertex,draw=KITgreen,fill=KITgreen!15] {\gs};%
		\node (tc1_v) at (tc1)  [vertex,draw=KITgreen,fill=KITgreen!15] {\gs};%
        \node (td0_v) at (td0)  [vertex,draw=KITred,fill=KITred!15] {\gs};%
        \node (td1_v) at (td1)  [vertex,draw=KITred,fill=KITred!15] {\gs};%
        
        \begin{pgfonlayer}{background}
        \draw [KITblue, route] (ta0) -- (ta1_v);
        \draw [KITblue, route] (tb0) -- (tb1_v);
        \draw [KITgreen, route] (tc0) -- (tc1_v);
        \draw [KITred, route] (td0) -- (td1_v);
                
        \draw [directedEdge]  (ta1) -- (tc0_v);
        \draw [directedEdge]  (tb1) -- (td0_v);
        \end{pgfonlayer}
		
		\node at (ta0) [text=KITblue] {\small{$0$}};%
		\node at (tb0) [text=KITblue] {\small{$0$}};%
		\node at (ta1) [text=KITblue] {\small{$1$}};%
		\node at (tb1) [text=KITblue] {\small{$1$}};%
		\node at (tc0)  [text=KITgreen] {\small{$0$}};%
		\node at (tc1)  [text=KITgreen] {\small{$1$}};%
        \node at (td0)  [text=KITred] {\small{$0$}};%
        \node at (td1)  [text=KITred] {\small{$1$}};%
		\node at (source_label) [text=nodeColor!100] {\small{$\sourceStop$}};%
		\node at (target_label) [text=nodeColor!100] {\small{$\targetStop$}};%
	\end{scope}
\end{tikzpicture}

%% file: fig/layoutGraphLeft.tex
\begin{tikzpicture}[scale=1.04]
	\begin{scope}[yscale=0.7]
		\node (v1) at (0.00, 0.00) {};%
		\node (v2) at (2.00, 0.00) {};%
		\node (v3) at (4.00, 0.00) {};%
		\node (v5) at (2.00, -1.50) {};%
		\node (v7) at (0.00, -1.50) {};%
		\node (v8) at (4.00, -1.50) {};%
		\node (v8up) at (4.00, -1.40) {};%
		\node (v8down) at (4.00, -1.60) {};%
		
		\node (v1v) at (v1) [vertex,draw=nodeColor,fill=nodeColor!15] {\gs};%
		\node (v2v) at (v2) [vertex,draw=nodeColor,fill=nodeColor!15] {\gs};%
		\node (v3v) at (v3) [vertex,draw=nodeColor,fill=nodeColor!15] {\gs};%
		\node (v5v) at (v5) [vertex,draw=nodeColor,fill=nodeColor!15] {\gs};%
		\node (v7v) at (v7) [vertex,draw=nodeColor,fill=nodeColor!15] {\gs};%
		\node (v8v) at (v8) [vertex,draw=nodeColor,fill=nodeColor!15] {\gs};%
        \node (v8upv) at (v8up) [vertex,draw=none,fill=none] {\gs};%
        \node (v8downv) at (v8down) [vertex,draw=none,fill=none] {\gs};%
        
        \begin{pgfonlayer}{background}
        \draw [KITgreen, route] (v1) -- (v2v);
        \draw [KITgreen, route] (v2) -- (v3v);
        \draw [KITblue, route] (v7) -- (v5v);
        \draw [KITblue, route, shorten <= -1pt] (v5v.330) -- (v8v.210);
        \draw [KITred, route, shorten <= -1pt] (v2v.240) -- (v5v.120);
        \draw [KITlilac, route, shorten <= -1pt] (v2v.300) -- (v5v.60);
        \draw [KITlilac, route, shorten <= -1pt] (v5v.30) -- (v8v.150);
        \draw [KITlilac, route] (v8) -- (v3v);
        
        \draw [directedEdge]  (v1) to [bend right] (v7v);
        \draw [directedEdge]  (v7) to [bend right] (v1v);
        \end{pgfonlayer}
	\end{scope}
\end{tikzpicture}

%% file: fig/layoutGraphRight.tex
\begin{tikzpicture}[scale=1.04]
	\begin{scope}[yscale=0.7]
		\node (v1) at (0.00, 0.00) {};%
		\node (v2) at (2.00, 0.00) {};%
		\node (v3) at (4.00, 0.00) {};%
		\node (v5) at (2.00, -1.50) {};%
		\node (v7) at (0.00, -1.50) {};%
		\node (v8) at (4.00, -1.50) {};%
		
		\node (v1v) at (v1) [vertex,draw=nodeColor,fill=nodeColor!15] {\gs};%
		\node (v2v) at (v2) [vertex,draw=nodeColor,fill=nodeColor!15] {\gs};%
		\node (v3v) at (v3) [vertex,draw=nodeColor,fill=nodeColor!15] {\gs};%
		\node (v5v) at (v5) [vertex,draw=nodeColor,fill=nodeColor!15] {\gs};%
		\node (v7v) at (v7) [vertex,draw=nodeColor,fill=nodeColor!15] {\gs};%
		\node (v8v) at (v8) [vertex,draw=nodeColor,fill=nodeColor!15] {\gs};%

        \begin{pgfonlayer}{background}
        \node [fit=(v1)(v7),line width=.5pt, fill=KITred!15,rounded corners=0.1cm] {};%
		\node [fit=(v2)(v3),line width=.5pt, fill=KITgreen!15,rounded corners=0.1cm] {};%
		\node [fit=(v5)(v8),line width=.5pt, fill=KITblue!15,rounded corners=0.1cm] {};%
        
        \path [directedEdge]  (v1v) edge node [above] {1} (v2v);
		\path [directedEdge,bend right]  (v1) edge node [left] {1} (v7v);
        \path [directedEdge,bend right]  (v7) edge node [right] {1} (v1v);
		\path [directedEdge]  (v2v) edge node [above] {1} (v3v);
		\path [directedEdge]  (v2v) edge node [left] {2} (v5v);
		\path [directedEdge]  (v8v) edge node [right] {1} (v3v);
		\path [directedEdge]  (v5v) edge node [above] {1} (v7v);
		\path [directedEdge]  (v5v) edge node [above] {2} (v8v);
        \end{pgfonlayer}
	\end{scope}
\end{tikzpicture}

%% file: fig/timeExpandedGraph.tex
\begin{tikzpicture}[xscale=0.8]
\begin{scope}[yscale=0.7]
\colorlet{arrivalColor}{KITlilac}
\colorlet{transferColor}{KITblack!80}
\colorlet{departureColor}{KITorange}

\node (arrival0Stop0) at (0.00, 0.00) {};%
\node (arrival1Stop0) at (0.00, -1.50) {};%
\node (arrival2Stop0) at (0.00, -3.00) {};%
\node (transfer0Stop0) at (2.00, 0.50) {};%
\node (transfer1Stop0) at (2.00, -1.00) {};%
\node (transfer2Stop0) at (2.00, -2.50) {};%
\node (departure0Stop0) at (4.00, 0.00) {};%
\node (departure1Stop0) at (4.00, -1.50) {};%
\node (departure2Stop0) at (4.00, -3.00) {};%

\node (arrival0Stop1) at (8.00, 0.00) {};%
\node (arrival1Stop1) at (8.00, -1.50) {};%
\node (arrival2Stop1) at (8.00, -3.00) {};%
\node (transfer0Stop1) at (10.00, 0.50) {};%
\node (transfer1Stop1) at (10.00, -1.00) {};%
\node (transfer2Stop1) at (10.00, -2.50) {};%
\node (departure0Stop1) at (12.00, 0.00) {};%
\node (departure1Stop1) at (12.00, -1.50) {};%
\node (departure2Stop1) at (12.00, -3.00) {};%

\node (departure0StopNEG1) at (-2.00, 0.00) {};%
\node (departure1StopNEG1) at (-2.00, -1.50) {};%
\node (departure2StopNEG1) at (-2.00, -3.00) {};%
\node (arrival0Stop2) at (14.00, 0.00) {};%
\node (arrival1Stop2) at (14.00, -1.50) {};%
\node (arrival2Stop2) at (14.00, -3.00) {};%
\node (transferNEG1Stop0) at (2.00, 1.50) {};%
\node (transferNEG1Stop1) at (10.00, 1.50) {};%
\node (transfer3Stop0) at (2.00, -4.00) {};%
\node (transfer3Stop1) at (10.00, -4.00) {};%

\begin{pgfonlayer}{background}
\node [stop,fit=(arrival0Stop0)(transfer0Stop0)(departure0Stop0)(arrival1Stop0)(transfer1Stop0)(departure1Stop0)(arrival2Stop0)(transfer2Stop0)(departure2Stop0)(transferNEG1Stop0)(transfer3Stop0)] {};
\node [stop,fit=(arrival0Stop1)(transfer0Stop1)(departure0Stop1)(arrival1Stop1)(transfer1Stop1)(departure1Stop1)(arrival2Stop1)(transfer2Stop1)(departure2Stop1)(transferNEG1Stop1)(transfer3Stop1)] {};
\end{pgfonlayer}

\node [align=left,text=KITblue] at (6.00, 0.30) {$\tripA$};%
\node [align=left,text=KITgreen] at (6.00, -1.20) {$\tripB$};%
\node [align=left,text=KITred] at (6.00, -2.70) {$\tripC$};%

\node (arrival0Stop0_v) at (arrival0Stop0) [vertex,draw=arrivalColor,fill=arrivalColor!50] {\gs};%
\node (arrival1Stop0_v) at (arrival1Stop0) [vertex,draw=arrivalColor,fill=arrivalColor!50] {\gs};%
\node (arrival2Stop0_v) at (arrival2Stop0) [vertex,draw=arrivalColor,fill=arrivalColor!50] {\gs};%
\node (arrival0Stop1_v) at (arrival0Stop1) [vertex,draw=arrivalColor,fill=arrivalColor!50] {\gs};%
\node (arrival1Stop1_v) at (arrival1Stop1) [vertex,draw=arrivalColor,fill=arrivalColor!50] {\gs};%
\node (arrival2Stop1_v) at (arrival2Stop1) [vertex,draw=arrivalColor,fill=arrivalColor!50] {\gs};%

\node (transfer0Stop0_v) at (transfer0Stop0) [vertex,draw=transferColor,fill=transferColor!50] {\gs};%
\node (transfer1Stop0_v) at (transfer1Stop0) [vertex,draw=transferColor,fill=transferColor!50] {\gs};%
\node (transfer2Stop0_v) at (transfer2Stop0) [vertex,draw=transferColor,fill=transferColor!50] {\gs};%
\node (transfer0Stop1_v) at (transfer0Stop1) [vertex,draw=transferColor,fill=transferColor!50] {\gs};%
\node (transfer1Stop1_v) at (transfer1Stop1) [vertex,draw=transferColor,fill=transferColor!50] {\gs};%
\node (transfer2Stop1_v) at (transfer2Stop1) [vertex,draw=transferColor,fill=transferColor!50] {\gs};%

\node (departure0Stop0_v) at (departure0Stop0) [vertex,draw=departureColor,fill=departureColor!50] {\gs};%
\node (departure1Stop0_v) at (departure1Stop0) [vertex,draw=departureColor,fill=departureColor!50] {\gs};%
\node (departure2Stop0_v) at (departure2Stop0) [vertex,draw=departureColor,fill=departureColor!50] {\gs};%
\node (departure0Stop1_v) at (departure0Stop1) [vertex,draw=departureColor,fill=departureColor!50] {\gs};%
\node (departure1Stop1_v) at (departure1Stop1) [vertex,draw=departureColor,fill=departureColor!50] {\gs};%
\node (departure2Stop1_v) at (departure2Stop1) [vertex,draw=departureColor,fill=departureColor!50] {\gs};%

\begin{pgfonlayer}{background}
\draw [edge, arrow] (arrival0Stop0) to[bend right=10] (transfer1Stop0_v);
\draw [edge, arrow] (transfer0Stop0) to[bend right=10] (departure0Stop0_v);
\draw [edge, arrow] (transferNEG1Stop0) -- (transfer0Stop0_v);
\draw [edge, arrow] (transfer0Stop0) -- (transfer1Stop0_v);
\draw [edge, arrow] (arrival1Stop0) to[bend right=10] (transfer2Stop0_v);
\draw [edge, arrow] (transfer1Stop0) to[bend right=10] (departure1Stop0_v);
\draw [edge, arrow] (transfer1Stop0) -- (transfer2Stop0_v);
\draw [edge, arrow, dashed] (transfer2Stop0) -- (transfer3Stop0);
\draw [edge, arrow] (transfer2Stop0) to[bend right=10] (departure2Stop0_v);
\draw [edge, arrow, dashed] (arrival2Stop0) to[bend right=10] (transfer3Stop0);

\draw [edge, arrow] (arrival0Stop1) to[bend right=10] (transfer1Stop1_v);
\draw [edge, arrow] (transfer0Stop1) to[bend right=10] (departure0Stop1_v);
\draw [edge, arrow] (transferNEG1Stop1) -- (transfer0Stop1_v);
\draw [edge, arrow] (transfer0Stop1) -- (transfer1Stop1_v);
\draw [edge, arrow] (arrival1Stop1) to[bend right=10] (transfer2Stop1_v);
\draw [edge, arrow] (transfer1Stop1) to[bend right=10] (departure1Stop1_v);
\draw [edge, arrow] (transfer1Stop1) -- (transfer2Stop1_v);
\draw [edge, arrow, dashed] (transfer2Stop1) -- (transfer3Stop1);
\draw [edge, arrow] (transfer2Stop1) to[bend right=10] (departure2Stop1_v);
\draw [edge, arrow, dashed] (arrival2Stop1) to[bend right=10] (transfer3Stop1);

\draw [route, KITblue] (departure0StopNEG1) -- (arrival0Stop0_v);
\draw [route, KITblue] (departure0Stop0) -- (arrival0Stop1_v);
\draw [route, KITblue] (departure0Stop1) -- (arrival0Stop2);
\draw [route, KITblue, shorten <= -0.2pt, shorten >= -0.2pt] (arrival0Stop0_v.50) arc (140:40:2.375cm);
\draw [route, KITblue, shorten <= -0.2pt, shorten >= -0.2pt] (arrival0Stop1_v.50) arc (140:40:2.375cm);

\draw [route, KITgreen] (departure1StopNEG1) -- (arrival1Stop0_v);
\draw [route, KITgreen] (departure1Stop0) -- (arrival1Stop1_v);
\draw [route, KITgreen] (departure1Stop1) -- (arrival1Stop2);
\draw [route, KITgreen, shorten <= -0.2pt, shorten >= -0.2pt] (arrival1Stop0_v.50) arc (140:40:2.375cm);
\draw [route, KITgreen, shorten <= -0.2pt, shorten >= -0.2pt] (arrival1Stop1_v.50) arc (140:40:2.375cm);

\draw [route, KITred] (departure2StopNEG1) -- (arrival2Stop0_v);
\draw [route, KITred] (departure2Stop0) -- (arrival2Stop1_v);
\draw [route, KITred] (departure2Stop1) -- (arrival2Stop2);
\draw [route, KITred, shorten <= -0.2pt, shorten >= -0.2pt] (arrival2Stop0_v.50) arc (140:40:2.375cm);
\draw [route, KITred, shorten <= -0.2pt, shorten >= -0.2pt] (arrival2Stop1_v.50) arc (140:40:2.375cm);
\end{pgfonlayer}

\node at (arrival0Stop0) [text=nodeColor!100] {\small{$a$}};%
\node at (arrival1Stop0) [text=nodeColor!100] {\small{$a$}};%
\node at (arrival2Stop0) [text=nodeColor!100] {\small{$a$}};%

\node at (transfer0Stop0) [text=nodeColor!100] {\small{$t$}};%
\node at (transfer1Stop0) [text=nodeColor!100] {\small{$t$}};%
\node at (transfer2Stop0) [text=nodeColor!100] {\small{$t$}};%

\node at (departure0Stop0) [text=nodeColor!100] {\small{$d$}};%
\node at (departure1Stop0) [text=nodeColor!100] {\small{$d$}};%
\node at (departure2Stop0) [text=nodeColor!100] {\small{$d$}};%

\node at (arrival0Stop1) [text=nodeColor!100] {\small{$a$}};%
\node at (arrival1Stop1) [text=nodeColor!100] {\small{$a$}};%
\node at (arrival2Stop1) [text=nodeColor!100] {\small{$a$}};%

\node at (transfer0Stop1) [text=nodeColor!100] {\small{$t$}};%
\node at (transfer1Stop1) [text=nodeColor!100] {\small{$t$}};%
\node at (transfer2Stop1) [text=nodeColor!100] {\small{$t$}};%

\node at (departure0Stop1) [text=nodeColor!100] {\small{$d$}};%
\node at (departure1Stop1) [text=nodeColor!100] {\small{$d$}};%
\node at (departure2Stop1) [text=nodeColor!100] {\small{$d$}};%

\end{scope}
\end{tikzpicture}

%% file: fig/speedupPlot.tex
\newcommand{\plotHeight}{6.3cm}
\newcommand{\plotWidth}{0.48\textwidth}

\begin{tikzpicture}
	\begin{groupplot}[group style={group size= 2 by 2, vertical sep=60, horizontal sep=60}]
		\nextgroupplot[
            title=\textbf{Paris -- Speedup},
            height=\plotHeight,
            width=\plotWidth,
            xlabel={$\numCells$},
            ylabel={Speedup},
            symbolic x coords={2, 4, 8, 16, 32, 64, 128, 256, 512, 1024, 2048, 4096, 8192, 16384, 32768},
			xticklabels={$2^1$, $ $, $2^3$, $ $, $2^5$, $ $, $2^7$, $ $, $2^9$, $ $, $2^{11}$, $ $, $2^{13}$, $ $, $2^{15}$},
            xtick=data,
            ymin=0,
            ymax=125,
            bar width = 4pt
        ]
		\addplot[ybar, /pgf/bar shift=-2.5pt, fill=KITred, KITred, thick] table [x index=0, y index=1, col sep=comma] {csv/paris.csv};
		\addplot[ybar, /pgf/bar shift=2.5pt, fill=KITblue, KITblue, thick] table [x index=0, y index=2, col sep=comma] {csv/paris.csv};
		
		\nextgroupplot[
            title=\textbf{Paris -- Memory},
            height=\plotHeight,
            width=\plotWidth,
            xlabel={$\numCells$},
            ylabel={Memory $[\mathrm{GB}]$},
            symbolic x coords={2, 4, 8, 16, 32, 64, 128, 256, 512, 1024, 2048, 4096, 8192, 16384, 32768},
			xticklabels={$2^1$, $ $, $2^3$, $ $, $2^5$, $ $, $2^7$, $ $, $2^9$, $ $, $2^{11}$, $ $, $2^{13}$, $ $, $2^{15}$},
            xtick=data,
            ymin=0,
            ymax=32,
            bar width = 4pt
        ]
		\addplot[ybar, /pgf/bar shift=-2.5pt, fill=KITred, KITred, thick] table [x index=0, y index=3, col sep=comma] {csv/paris.csv};
		\addplot[ybar, /pgf/bar shift=2.5pt, fill=KITblue, KITblue, thick] table [x index=0, y index=4, col sep=comma] {csv/paris.csv};
		
		\nextgroupplot[
			name=mainPlot,
            title=\textbf{Switzerland -- Speedup},
            height=\plotHeight,
            width=\plotWidth,
            xlabel={$\numCells$},
            ylabel={Speedup},
            symbolic x coords={2, 4, 8, 16, 32, 64, 128, 256, 512, 1024, 2048, 4096, 8192, 16384, 30510},
			xticklabels={$2^1$, $$, $2^3$, $$, $2^5$, $$, $2^7$, $$, $2^9$, $$, $2^{11}$, $$, $2^{13}$, $$, $\absoluteVal{\stops}$},
            xtick=data,
            ymin=0,
            ymax=300,
            bar width = 4pt
        ]
		\addplot[ybar, /pgf/bar shift=-2.5pt, fill=KITred, KITred, thick] table [x index=0, y index=1, col sep=comma] {csv/switzerland.csv};
		\addplot[ybar, /pgf/bar shift=2.5pt, fill=KITblue, KITblue, thick] table [x index=0, y index=2, col sep=comma] {csv/switzerland.csv};
		\coordinate (c1) at (rel axis cs:0,1);
		
		\nextgroupplot[
            title=\textbf{Switzerland -- Memory},
            height=\plotHeight,
            width=\plotWidth,
            xlabel={$\numCells$},
            ylabel={Memory $[\mathrm{GB}]$},
            legend to name=sweden,
            symbolic x coords={2, 4, 8, 16, 32, 64, 128, 256, 512, 1024, 2048, 4096, 8192, 16384, 30510},
            xticklabels={$2^1$, $$, $2^3$, $$, $2^5$, $$, $2^7$, $$, $2^9$, $$, $2^{11}$, $$, $2^{13}$, $$, $\absoluteVal{\stops}$},
            xtick=data,
            ymin=0,
            ymax=15,
            bar width = 4pt,
            legend image code/.code={\draw [#1] (0cm,-0.1cm) rectangle (0.5cm,0.1cm);}
        ]
		\addplot[ybar, /pgf/bar shift=-2.5pt, fill=KITred, KITred, thick] table [x index=0, y index=3, col sep=comma] {csv/switzerland.csv};\label{legend:uncomp}
		\addplot[ybar, /pgf/bar shift=2.5pt, fill=KITblue, KITblue, thick] table [x index=0, y index=4, col sep=comma] {csv/switzerland.csv};\label{legend:comp}
		\coordinate (c2) at (rel axis cs:1,1);
	\end{groupplot}
	
	\arrayrulecolor{legendColor}
	\node[inner sep=0pt,outer sep=0pt] (legend) at (mainPlot.outer south west) {};
	\node[inner sep=0pt,outer sep=0pt,anchor=north west] at (legend) {
		\small
		\begin{tabular*}{\textwidth}{|@{~~}l@{\extracolsep{\fill}}r@{~~}|}
			\hline
			& \\[-7pt]
			\legend{\ref{legend:uncomp}}~FLASH-TB (without flag compression) & \legend{\ref{legend:comp}}~FLASH-TB (with flag compression) \\[2pt]
			\hline
		\end{tabular*}
	};
	\arrayrulecolor{black}
\end{tikzpicture}

%% file: fig/flagDistribution.tex
\newcommand{\plotHeight}{6.3cm}
\newcommand{\plotWidth}{0.48\textwidth}

\begin{tikzpicture}
	\begin{groupplot}[group style={group size= 2 by 2, vertical sep=60, horizontal sep=60}]
		\nextgroupplot[
            title=\textbf{Germany},
            height=\plotHeight,
            width=\plotWidth,
            xlabel={Number of set flags},
            ylabel={Transfers [\%]},
            scaled y ticks=false,
            xmin=-0.5,
            xmax=32.5,
            ymin=0,
            ymax=60,
            bar width = 5pt
        ]
		\addplot[ybar, fill=KITblue50, draw=KITblue] table [x index=0, y index=1, col sep=tab] {csv/new_distribution_32.tsv};
		
		\nextgroupplot[
		title=\textbf{Paris},
		height=\plotHeight,
		width=\plotWidth,
		xlabel={Number of set flags},
		ylabel={Transfers [\%]},
		xmin=-0.5,
		xmax=32.5,
		ymin=0,
		ymax=60,
		bar width = 5pt,
		]
		\addplot[ybar, fill=KITblue50, draw=KITblue] table [x index=0, y index=3, col sep=tab] {csv/new_distribution_32.tsv};

		\nextgroupplot[
            title=\textbf{Sweden},
            height=\plotHeight,
            width=\plotWidth,
            xlabel={Number of set flags},
            ylabel={Transfers [\%]},
            xmin=-0.5,
			xmax=32.5,
            ymin=0,
            ymax=60,
            bar width = 5pt
        ]
		\addplot[ybar, fill=KITblue50, draw=KITblue] table [x index=0, y index=4, col sep=tab] {csv/new_distribution_32.tsv};
		
		\nextgroupplot[
            title=\textbf{Switzerland},
            height=\plotHeight,
            width=\plotWidth,
            xlabel={Number of set flags},
            ylabel={Transfers [\%]},
            xmin=-0.5,
			xmax=32.5,
            ymin=0,
            ymax=60,
            bar width = 5pt
        ]
		\addplot[ybar, fill=KITblue50, draw=KITblue] table [x index=0, y index=2, col sep=tab] {csv/new_distribution_32.tsv};
	\end{groupplot}
\end{tikzpicture}